\documentclass{llncs}
\pdfoutput=1
\usepackage{wrapfig}
\usepackage{latexsym}
\usepackage{times}
\usepackage{graphicx}
\usepackage{amsmath,amssymb,amsfonts}
\usepackage{color}
\usepackage{textcomp}
\usepackage{tabularx}
\usepackage{times}

\graphicspath{{figures/}}
\newcommand{\remove}[1]{}

\begin{document}

\title{Straight-line Drawability of a \\ Planar Graph Plus an Edge\thanks
{This research began at the {\em Blue Mountains Workshop on
Geometric Graph Theory}, August, 2010, in Australia,
and supported by the University of Sydney IPDF funding and the ARC (Australian Research Council).
Hong is supported by ARC Future Fellowship.
Liotta is also supported  by the Italian
Ministry of Education, University, and Research (MIUR) under PRIN
2012C4E3KT AMANDA.}}

\author{
P. Eades\inst{1} \and S.H. Hong\inst{1} \and G. Liotta\inst{2} \and
N. Katoh\inst{3} \and S.H. Poon\inst{4}}

\institute{
University of Sydney, Australia
\email{\{peter.eades,seokhee.hong\}@sydney.edu.au}
\ \\
\and
University of Perugia, Italy
\email{liotta@diei.unipg.it}
\ \\
\and
Kyoto University, Japan
\email{naoki@archi.kyoto-u.ac.jp}
\ \\
\and
National Tsing Hua University, Taiwan
\email{spoon@cs.nthu.edu.tw}
}

\date{}
\maketitle

%\vspace{2cm}
%\framebox[1.1\width]{\today} \par
%\vspace{2cm}

\begin{abstract}
We investigate straight-line drawings of topological
graphs that consist of a planar graph plus one edge, also called
almost-planar graphs. We present a characterization of such graphs
that admit a straight-line drawing. The characterization enables a
linear-time testing algorithm to determine whether an almost-planar graph admits
a straight-line drawing, and a linear-time drawing algorithm that
constructs such a drawing, if it exists. We also show that some
almost-planar graphs require exponential area for a straight-line
drawing.
\end{abstract}

\section{Introduction}
\label{se:introduction}

This paper investigates straight-line drawings of \emph{almost-planar} graphs, that is,
graphs that become planar after the deletion of just one edge.

Our work is partly motivated by the classical \emph{planarization} approach~\cite{DETT} to graph drawing.
This method takes
as input a graph $G$, deletes a small number of edges to give a
planar subgraph $G^-$, and then constructs a planar topological embedding
(i.e., a plane graph) of $G^-$.
Then the deleted edges are re-inserted, one by one, to
give a topological embedding of the original graph $G$. Finally, a
drawing algorithm is applied to the topological embedding. A number
of variations on this basic approach give a number of graph drawing
algorithms (see, e.g.,~\cite{DETT}). This paper is concerned with the final
step of creating a drawing from the topological embedding.

Minimizing the number of edge crossings is
an NP-hard problem even when the given graph is almost-planar~\cite{CM2}.
However, Gutwenger, Mutzel, and Weiskircher~\cite{GMW} present an elegant
polynomial-time solution to the following simpler problem: Given a
graph $G$ and an edge $e$ such that $G^- = G - e$ is planar, find a
planar topological embedding of $G^-$ that minimizes the number of
edge crossings when re-inserting $e$ in $G$.

While the output of the algorithm of Gutwenger \emph{et al.} has the
minimum number of edge crossings, it may not give rise to a
straight-line planar drawing. In this paper we study the following
problem: Let $G$ be a topological graph consisting of a planar graph
plus an edge $e$. We want to test whether $G$ admits a straight-line
drawing that preserves the given embedding.

\begin{figure}[t!]
  \centering
  \includegraphics[width=.7\columnwidth]{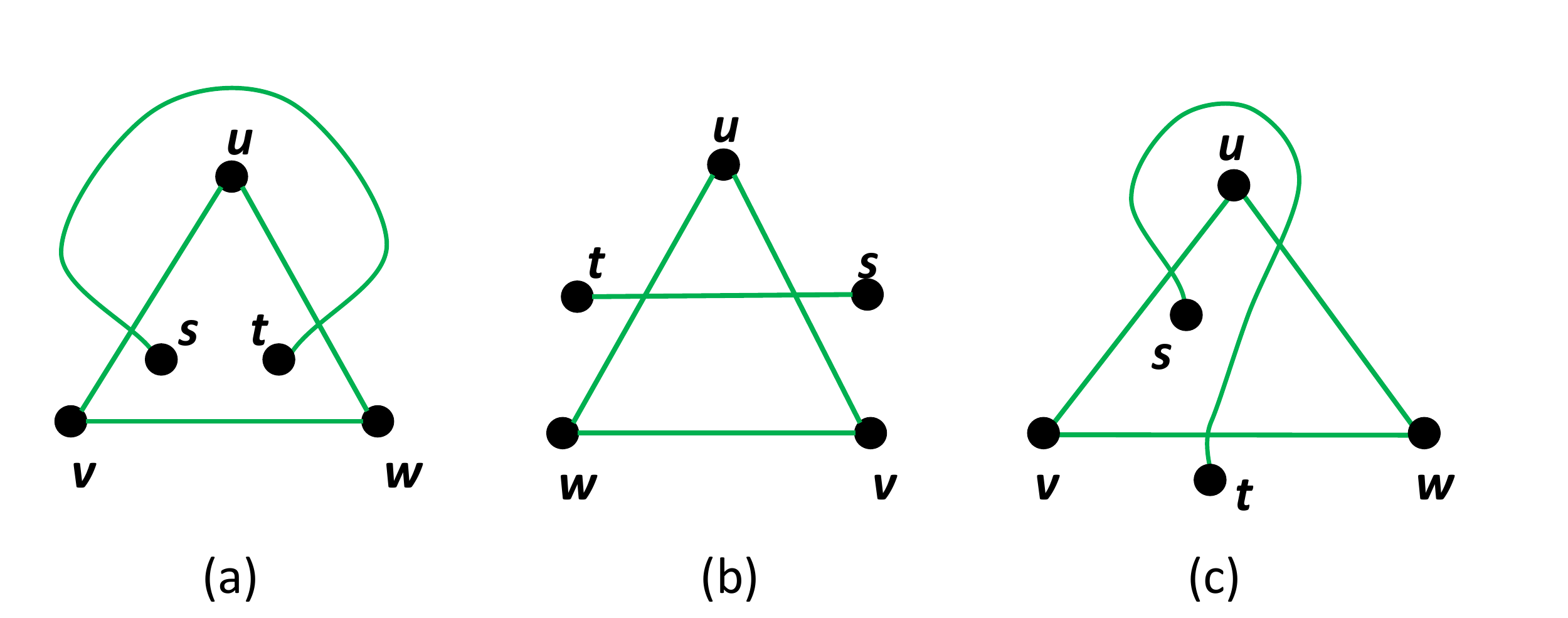}
  %\vspace{-1.5cm}
  \caption{
  (a) An almost-planar topological graph $G$; (b) a straight-line drawing of $G$ that preserves its embedding on the sphere but not on the plane; (c) An almost-planar topological graph for which an embedding preserving straight-line drawing does not exist.}
  \label{fi:intro}
\end{figure}

It is important to remark that by ``preserving the embedding'' we
mean that the straight-line drawing must preserve the cyclic order
of the edges around each vertex and around each crossing. In other
words, we want to preserve a given embedding on the {\em sphere}.
Note that the problem is different if, in addition to preserving the
cyclic order of the edges around the vertices and the crossings, we
also want the preservation of a given external boundary; in other
words the problem is different if we want to maintain a given
embedding on the {\em plane} instead of on the sphere. For example,
consider the graph of Fig.~\ref{fi:intro}(a). If we regard this as
a topological graph on the sphere, then it has an embedding
preserving straight-line drawing, as shown in
Fig.~\ref{fi:intro}(b). However, the drawing in
Fig.~\ref{fi:intro}(a) has a different external face to
Fig.~\ref{fi:intro}(b). It is easy to show that there is no
straight-line drawing with the same external face as in
Fig.~\ref{fi:intro}(a). For a contrast, Fig.~\ref{fi:intro}(c) shows
a topological graph $G$ that does not have a straight-line drawing
that preserves the embedding on the sphere.

In this paper we mostly focus on spherical topologies but, as a
byproduct, we obtain a result for topologies on the plane that may
be of independent interest. Namely, the main results of this paper
are as follows.

\begin{itemize}

\item We characterize those almost-planar topological graphs
that admit a straight-line drawing that preserves
a given embedding on the sphere. The characterization
gives rise to a linear-time testing algorithm.

\item We characterize those almost-planar topological graphs
that admit a straight-line drawing that preserves
a given embedding on the plane.

\item We present a drawing algorithm that
constructs straight-line drawings when such drawings exist. This
drawing algorithm runs in linear time; however, the model of
computation used is the real RAM, and the drawings that are produced
have exponentially bad resolution. We show that, in the worst case,
the exponentially bad resolution is inevitable.

\end{itemize}

Our results also contribute to the rapidly increasing literature
about topological graphs that are ``nearly'' plane, in some sense.
An interesting example is the class of \emph{1-plane graphs}, that
is, topological graphs with at most one crossing per edge.
Thomassen~\cite{Thomassen2} gives a ``F\'{a}ry-type theorem'' for
1-plane graphs, that is, a characterization of 1-plane topological
graphs that admit a straight-line drawing. Hong et al.~\cite{HELP}
present a linear-time algorithm  that constructs a straight-line
1-planar drawing of 1-plane graph, if it exists. More generally,
Nagamochi~\cite{Nagamochi} investigates straight-line drawability of
a wide class of topological non-planar topological graphs. He
presentes F\'{a}ry-type theorems as well as polynomial-time testing
and drawing algorithms. This paper considers graphs that are
``nearly plane'' in the sense that deletion of a single edge yields
a planar graph. Such graphs are variously called ``1-skew graphs''
or ``almost-planar'' graphs in the literature. Our characterization
can be regarded as a F\'{a}ry-type theorem for almost-planar graphs.

The rest of the paper is organized as follows.
Section~\ref{se:preliminaries} gives notation and terminology. The
characterization of almost-planar topological graphs on the sphere
that admit an embedding preserving straight-line drawing is given in
Section~\ref{se:s2}. The extension of this characterization to
topological graphs on the plane and the exponential area lower bound
are described in Section~\ref{se:wrap-up}. Open problems can be
found in Section~\ref{se:open}.

\section{Preliminaries}
\label{se:preliminaries}

%\subsection{Topological graphs and embeddings}

A \emph{topological graph} $G = (V, E)$ is a representation of a
simple graph on a given surface, where each vertex is represented by
a point and each edge is represented by a simple Jordan arc between
the points representing its endpoints. If the given surface is the
sphere, then we say that $G$ is an {\em $\mathbb{S}^2$-topological
graph}; if the given surface is the plane, then we say that $G$ is
an {\em $\mathbb{R}^2$-topological graph}. Two edges of a
topological graph \emph{cross} if they have a point in common, other
than their endpoints. The point in common is called a
\emph{crossing}. We assume that a topological graph satisfies the
following non-degeneracy conditions: (i) an edge does not contain a vertex other
than its endpoints; (ii) edges must not meet tangentially; (iii) no
three edges share a crossing; and (iv) an edge does not cross an
incident edge.

 An \emph{$\mathbb{S}^2$-embedding} of a graph is an equivalence class
of $\mathbb{S}^2$-topological graphs under homeomorphisms of the
sphere. An $\mathbb{S}^2$-topological graph has no unbounded face;
in fact an $\mathbb{S}^2$-embedding is uniquely determined merely by
the clockwise order of edges around each vertex and each edge
crossing. An \emph{$\mathbb{R}^2$-embedding} of a graph is an
equivalence class of $\mathbb{R}^2$-topological graphs under
homeomorphisms of the plane. Note that one face of an
$\mathbb{R}^2$-topological graph in the plane is unbounded; this is
the \emph{external face}.

The concepts of $\mathbb{R}^2$-embedding and
$\mathbb{S}^2$-embedding are very closely related. Each
$\mathbb{S}^2$-topological graph gives rise to a representation of
the same graph on the plane, by a stereographic projection about
an interior point of a
chosen face. This chosen face becomes the external face of the
$\mathbb{R}^2$-topological graph. Thus we can
regard an $\mathbb{R}^2$-embedding to be an $\mathbb{S}^2$-embedding
in which one specific face is chosen to be the external face.
Further, each
$\mathbb{R}^2$-topological graph gives rise to a representation of
the same graph on the sphere, by a simple projection.

%\subsection{Planar and almost-planar embeddings}

A topological graph (either on the plane or on the sphere) is {\em planar} if no two edges cross.
A topological graph is {\em almost-planar}
if it has an edge $(s,t)$ whose removal makes it planar.
An \emph{almost-planar $\mathbb{R}^2$-embedding} ($\mathbb{S}^2$-embedding) of a graph is an equivalence class of almost-planar $\mathbb{R}^2$-topological graphs ($\mathbb{S}^2$-topological graphs)
under homeomorphisms of the plane (sphere).

Throughout this paper, $G = (V,E)$ denotes an almost-planar
topological graph ($\mathbb{S}^2$ or $\mathbb{R}^2$) and $(s,t)$
denotes an edge of $G$ whose deletion makes $G$ planar. The
embedding obtained by deleting the edge $(s,t)$ is denoted by
$\hat{G}$. More generally, we use the convention that the notation
$\hat{X}$ normally denotes $X$ without the edge $(s,t)$.

Let $G$ be an $\mathbb{S}^2$-topological graph and let $G'$ be an
$\mathbb{R}^2$-topological graph with the same underlying simple
graph. We say that $G'$ {\em preserves the $\mathbb{S}^2$-embedding
of $G$} if for each vertex and for each crossing they have the same
cyclic order of incident edges. Further, let $G$ be an
$\mathbb{R}^2$-topological graph and let $G'$ be an
$\mathbb{R}^2$-topological graph with the same underlying simple
graph. We say that $G'$ {\em preserves the $\mathbb{R}^2$-embedding
of $G$} if for each vertex and for each crossing they have the same
cyclic order of incident edges and the same external face. A
\emph{straight-line drawing} of a graph is an
$\mathbb{R}^2$-topological graph whose edges are represented by
straight-line segments.

\section{Straight-line drawability of an almost-planar $\mathbb{S}^2$-embedding}
\label{se:s2}

In this section we state our main theorem. Let $G$ be a topological
graph with a given almost-planar $\mathbb{S}^2$-embedding. Suppose
that $\alpha$ is a crossing between edges $(s,t)$ and $(u,v)$ in
$G$. If the clockwise order of vertices around $\alpha$ is $\langle
s, u , t , v \rangle$, then $u$ is a \emph{left} vertex and $v$ is a
\emph{right} vertex (with respect to the ordered pair $(s,t)$ and the crossing $\alpha$). We
say that a vertex of $G$
 is \emph{inconsistent} if it is both left and right, and
\emph{consistent} otherwise. For example, vertex $v$ in
Fig.~\ref{fi:intro}(c) is inconsistent: it is a left vertex
with respect the first crossing along $(s,t)$,
and it is a right vertex with respect to the final crossing along $(s,t)$.

\begin{theorem}\label{th:s2}
An almost-planar $\mathbb{S}^2$-topological graph $G$ with $n$
vertices admits an $\mathbb{S}^2$-embedding preserving straight-line
 drawing if and only if every
vertex of $G$ is consistent. This condition can be tested in $O(n)$
time.
\end{theorem}

The necessity of every vertex being consistent is straightforward.
The proof of sufficiency involves many technicalities and it
occupies most of the remainder of this paper. Namely, we prove the
sufficiency of the condition in Theorem~\ref{th:s2} by the following
steps.

\begin{description}
\item [Augmentation:] We show that we can add edges to an almost-planar $\mathbb{S}^2$-topological graph
to form a maximal almost-planar graph, without changing the property
that every vertex is consistent. Let $G'$ be the augmented
$\mathbb{S}^2$-topological graph (subsection~\ref{ss:augmentation}).

\item [Choice of an external face:] We find a face $f_o$ of $G'$
such that if the $\mathbb{S}^2$-embedding of $G'$ is projected on
the plane with $f_0$ as the external face, $G'$ satisfies an
additional property that we call {\em face consistency}
(subsection~\ref{ss:external-face}).

\item [Split the augmented graph:] After having projected
$G'$ on the plane with $f_o$ as the external face, we split  the
$\mathbb{R}^2$-embedding of $G'$ into the ``inner graph''   and the
``outer graph''. The inner graph and outer graph share a cycle
called the ``separating cycle'' (subsection~\ref{ss:divide}).

\item [Straight-line drawing computation:] We draw the outer graph leaving a convex shaped ``hole'' for the
inner graph; the boundary of this hole is the separating cycle. Then
we draw the ``inner graph'', whose external face is the separating
cycle, such that it fits exactly into the convex shaped ``hole''
(subsections~\ref{ss:outerGraph} and \ref{ss:innerGraph}).
\end{description}

Before presenting more details of the proof of sufficiency, we
observe that the condition stated in Theorem~\ref{th:s2} can be
tested in linear time. By regarding crossing points as dummy
vertices, we can apply the usual data structures for plane graphs to
almost-planar graphs (see~\cite{DBLP:journals/tcs/ChrobakE91}, for
example). A simple traversal of the crossing points along the edge
$(s,t)$ can be used to compute the left and the right vertices.
Since the number of crossing points in an almost-planar graph is
linear, these data structures can be applied without asymptotically
increasing total time complexity.

\subsection{Augmentation}\label{ss:augmentation}

Let $G$ be an $\mathbb{S}^2$-topological graph. An {\em
$\mathbb{S}^2$-embedding preserving augmentation} of $G$ is an
$\mathbb{S}^2$-topological graph $G'$ obtained by adding edges (and
no vertices) to $G$ such that for each vertex  (for each crossing)
of $G'$, the cyclic order of the edges of $G' \cap G$ around the
vertex (around the crossing) is the same in $G'$ and in $G$. An
almost-planar topological graph is \emph{maximal} if the addition of
any edge would result in a topological graph that is not almost
planar. The following lemma describes a technique to compute an
$\mathbb{S}^2$-embedding preserving augmentation of an almost-planar
$\mathbb{S}^2$-topological graph that gives rise to a maximal
almost-planar graph.

\begin{lemma}\label{le:augmentation}
Let $G$ be an almost-planar $\mathbb{S}^2$-topological graph with
$n$ vertices. If $G$ satisfies the vertex consistency condition,
then there exists a maximal almost-planar $\mathbb{S}^2$-embedding
preserving augmentation $G'$ of $G$ such that $G'$ satisfies the
vertex consistency condition. Also, such augmentation can be
computed in $O(n)$ time.
\end{lemma}

\begin{proof}
We add as many edges as possible to $G$ without introducing any new
crossing. Let $\tilde {G}$ be the resulting
$\mathbb{S}^2$-topological graph. Since $\tilde {G}$ was constructed
from $G$ without adding any edge crossings, $\tilde {G}$ is an
almost-planar $\mathbb{S}^2$-topological graph and it has the same
set of left vertices and the same set of right vertices as $G$. If
$\tilde {G} - (s,t)$ is a maximal planar graph we are done. So
assume $\tilde {G}$ otherwise, which implies that there is at least
one face of $\tilde {G} - (s,t)$ that has size larger than three.
Fig.~\ref{fi:Augmentation}(a) shows an example of an almost-planar
topological graph $G$ and Fig.~\ref{fi:Augmentation}(b) shows an
example of a (non-maximal) graph $\tilde {G}$.

\begin{figure}[h!]
 \centering
 \includegraphics[width=.9\columnwidth]{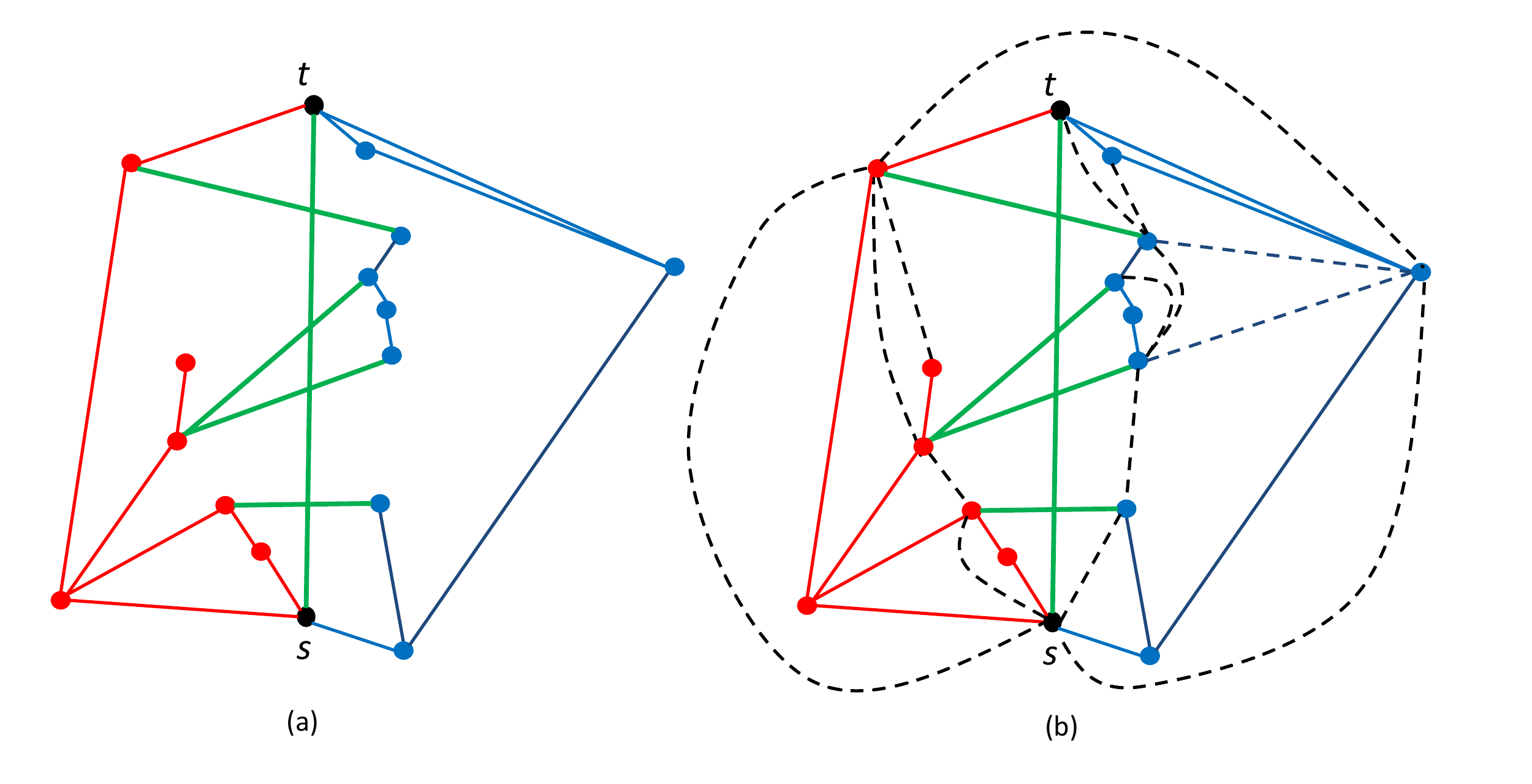}
 \caption{(a) A non maximal almost-planar graph $G$. (b) Graph $\tilde {G}$ obtained by adding the dotted edges.}
 \label{fi:Augmentation}
\end{figure}

Denote the edges of $G$ that cross $(s,t)$ by $e_0, e_1, \dots ,
e_{p-1}$, ordered from $s$ to $t$ by their crossings along $(s,t)$.
Suppose that $e_i = (\ell_i , r_i)$ for $0 \leq i < p$, where
$\ell_i$ is a left vertex and $r_i$ is a right vertex. By
construction, $\tilde {G}$ is such that $\ell_0$ is adjacent to $s$.
If not, we could have added edge $(\ell_0,s)$ to $\tilde {G}$ by
closely following the edge $e_0$ from $\ell_0$ to the crossing with
$(s,t)$ and then by closely following the edge $(s,t)$ from this
crossing down to $s$, till we encounter $s$, without introducing any
new edge crossings. With similar reasoning, it is immediate to see
that, in $\tilde {G}$, vertex $r_0$ is adjacent to $s$ and that
$\ell_{p-1}$ and $r_{p-1}$ are both adjacent to $t$. Also, any
two consecutive edges $e_i = (\ell_i , r_i)$ and $e_{i+1} =
(\ell_{i+1} , r_{i+1})$ either have an endvertex in common or
graph $\tilde {G}$ contains edges $(\ell_i, \ell_{i+1})$ and $(r_i,
r_{i+1})$ (or we could have added them without crossing edge $(s,t)$
by closely following $e_{i+1}$ from an endvertex to its crossing
with $(s,t)$, then closely following $(s,t)$ from this crossing down
to the crossing between $(s,t)$ and $e_i$, and finally closely
following $e_i$ to its endvertex).

Let $C_0$ be the cycle of $\tilde {G}$ induced by vertices $s$,
$\ell_0$, and $r_0$; let $C_{p-1}$ be the cycle induced by vertices
$t$, $\ell_{p-1}$, and $r_{p-1}$; for $1 \leq i < p-1$, let $C_i$ be
the cycle of $\tilde {G}$ induced by vertices $\ell_i$, $r_i$,
$\ell_{i+1}$, and $r_{i+1}$. Note that $C_0$ and $C_{p-1}$ are
3-cycles, while $C_i$ may be either a 3-cycle or a 4-cycle depending
on whether edges $e_i$ and $e_{i+1}$ have an endvertex in common or
not. See, for example, Fig.~\ref{fi:Augmentation}(b) where some of
such cycles are 3-cycles and some other such cycles are 4-cycles.
For each cycle $C_0, C_1,\dots , C_{p-1}$ we define the interior of
the cycle as the region that contains a portion of the edge $(s,t)$.

Consider $C_0$ first and the edge $(\ell_0,s)$. Since
$\tilde {G}$ is constructed by adding as many edges as possible to
$G$ that do not cross edge $(s,t)$, we have that $(\ell_0,s)$ is an
edge of some triangular face of $\tilde {G}$. Let $v$ be the vertex
of $f$ different from $\ell_0$ and $s$. Either $v$ is in the
interior of $C_0$ or not. Similarly, let $w$ be the vertex opposite
to $(r_0,s)$ in some triangular face of $\tilde {G}$; either $w$ is
in the interior of $C_0$ or not. We distinguish between three cases
(see Fig.~\ref{fi:three-cases-C0}): One of $v$ and $w$ is in the
interior of $C_0$, both $v$ and $w$ are in the interior of $C_0$,
and none of $v$ and $w$ is in the interior of $C_0$. Consider the
first case and assume, w.l.o.g., that $v$ is in the interior of
$C_0$ while $w$ is not. Note that $v$ was neither a left vertex nor
a right vertex of $G$; it cannot be a left vertex because there
would be an edge crossing $(s,t)$ that is encountered before $e_0$
when going from $s$ to $t$. It cannot be a right vertex, because the
edge incident to $v$ and crossing $(s,t)$ should also intersect edge
$(s, \ell_0)$, which is impossible.

 We add an edge connecting $v$ with $r_0$ as
described by the dotted edge in Fig.~\ref{fi:three-cases-C0}(a):
Start at $v$ and closely follow edge $(v, \ell_0)$, then closely
follow edge $(\ell_0, r_0)$ till $r_0$ is reached. Note that $v$ has
become a left vertex and that the vertex consistency condition holds
for all other vertices. Consider now the case that both $v$ and $w$
are in the interior of $C_0$ (see
Fig.~\ref{fi:three-cases-C0}(b)). Also in this case, both $v$ and
$w$ are neither left nor right vertices of $G$. We add edge
$(v,r_0)$ as in the previous case; then we add edge $(v,w)$ by
starting at $w$, closely following edge $(w,s)$, then closely
following edge $(s,v)$, and finally ending at vertex $v$. In this
case $v$ has become a left vertex, $w$ is a right vertex, and the
vertex consistency condition holds for all other vertices. Finally,
if neither $v$ nor $w$ is in the interior of $C_0$, no edge crossing
$(s,t)$ is added in its interior and the vertex condition
consistency is trivially maintained.

\begin{figure}[h!]
 \centering
 \includegraphics[width=.9\columnwidth]{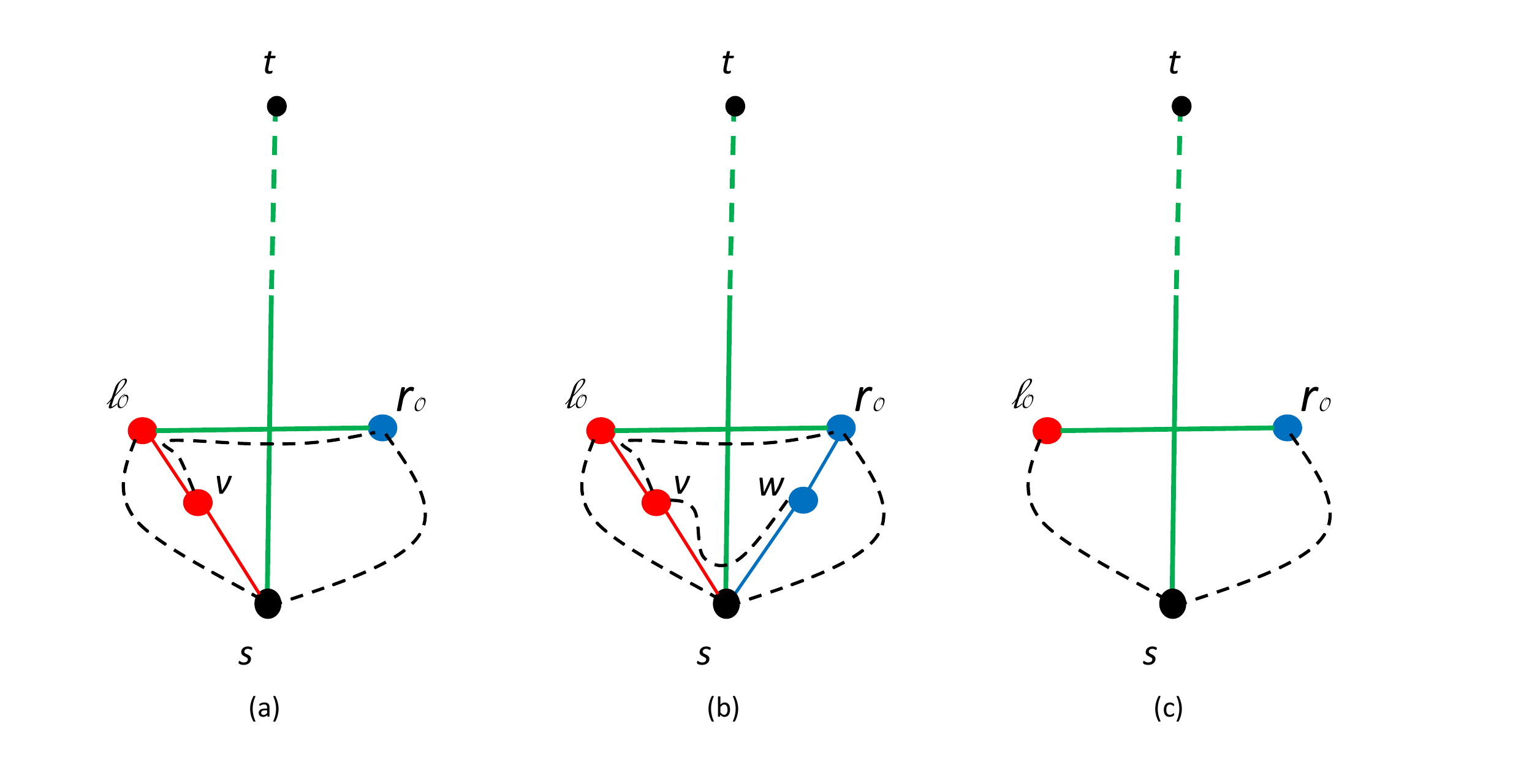}
 \caption{(a) The interior of $C_0$ contains a vertex $v$. (b) The interior of $C_0$ contains two vertices $v$ and $w$. (c) The interior of $C_0$ does not contain vertices.}
 \label{fi:three-cases-C0}
\end{figure}

For each cycle $C_i$ ($1 \leq i < p-1$) such that $C_i$ is not a
3-cycle we add an edge in its interior that only crosses edge
$(s,t)$ as follows: Starting at vertex $\ell_i$, we closely follow
edge $e_i$ until the crossing between $e_i$ and $(s,t)$ is
encountered; next, we closely follow edge $(s,t)$ upward to its
crossing with $e_{i+1}$; finally, we closely follow edge $e_{i+1}$
until we reach vertex $r_{i+1}$; see
Fig.~\ref{fi:adding-one-edge}(a) for an illustration of this step.
Note that since edge $(\ell_i,r_{i+1})$ is added in the interior of
$C_i$, we have that $\ell_i$ remains a left vertex and $r_{i+1}$
remains a right vertex after such an edge insertion. Consider now the
3-cycle with vertices $\ell_i$, $r_{i+1}$, and $\ell_{i+1}$ and the
the 3-cycle with vertices $r_i$, $r_{i+1}$, and $\ell{i}$. As with
the case of $C_0$, each such 3-cycle may have a vertex $v$ in its
interior such that $v$ is neither a left vertex nor a right vertex
of $G$. We add an edge incident to $v$ and crossing edge $(s,t)$ and
no other edges of the graph with the same technique described for
$C_0$; see Fig.~\ref{fi:adding-one-edge}(b) for an example of this
edge addition.

\begin{figure}[h!]
 \centering
 \includegraphics[width=.9\columnwidth]{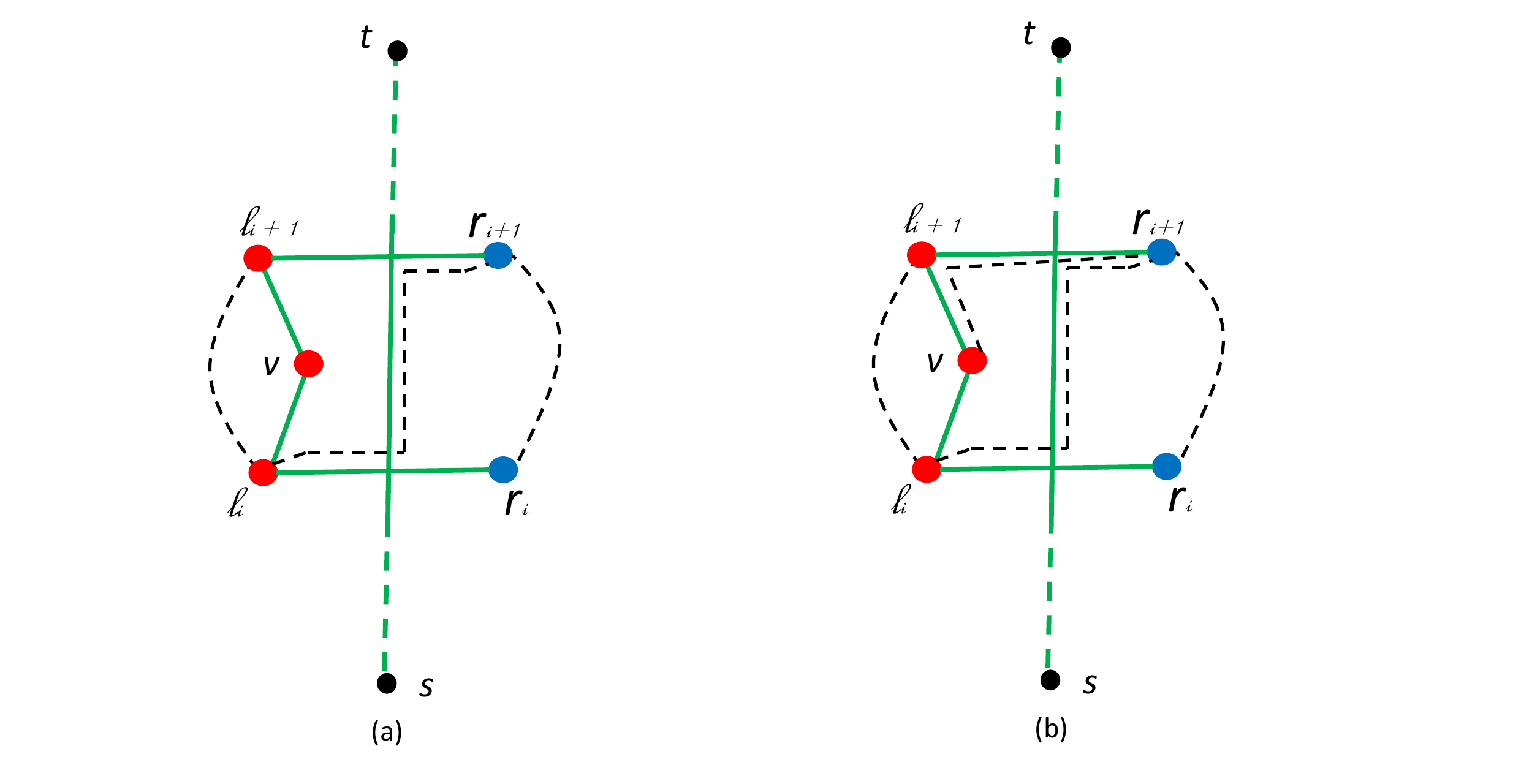}
 \caption{(a) $C_i$ is a 4-cycle and the dotted edge splits it into two 3-cycles. (b) A dotted edge crossing
 $(s,t)$ and incident to a vertex $v$ is added; $v$ becomes a left vertex.}
 \label{fi:adding-one-edge}
\end{figure}

Finally, edges are added inside cycle $C_{p-1}$ with the same
approach described for $C_0$.

Let $G'$ be the resulting $\mathbb{S}^2$-topological graph. The
proof is concluded by observing that: (i) $G'$ is an almost-planar
$\mathbb{S}^2$-embedding preserving augmentation $G$ and it has
$3n-5$ edges, and (ii) $G'$ can be constructed in linear time by a
simple traversal of the faces in a standard data structure for
planar graphs adapted to represent almost-planar graphs
(see~\cite{DBLP:journals/tcs/ChrobakE91}, for example). \qed
\end{proof}

It would be tempting to suggest that a maximal almost-planar
topological graph consists of a maximal planar graph plus an edge.
This is not quite true; see Fig~\ref{fi:example3n-6}.
\begin{figure}[ht!]
  \centering
  \includegraphics[width=.7\columnwidth]{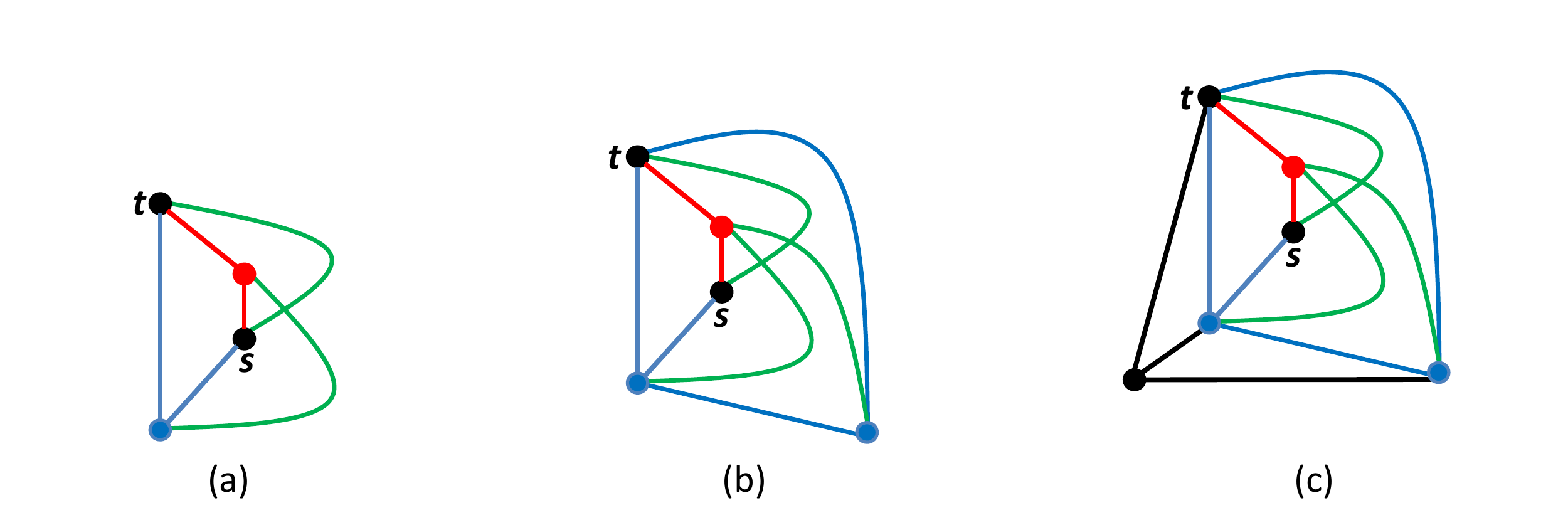}
  \caption{Maximal almost-planar graphs with $3n-6$ edges.}
  \label{fi:example3n-6}
\end{figure}
However, we can prove that a maximal almost-planar topological graph
``almost'' consists of a maximal planar graph plus an edge. To this
aim, we need a preliminary result.

\begin{lemma}
\label{le:faces} Suppose that $f$ is a face of the topological
subgraph $\hat{G}$ of a maximal almost-planar topological graph $G$
formed by deleting the edge $(s,t)$.
\begin{enumerate}
  \item [(1)] The subgraph $G_f$ of $G$ induced by the vertices of $f$ is a clique in $G$.
  \item [(2)] The subgraph $\hat{G}_f$ of $\hat{G}$ induced by the vertices of $f$ is an outerplane graph.
\end{enumerate}
\end{lemma}
\begin{proof}
Since $G$ is a maximal almost-planar graph, if $u$ and $v$ are
vertices of $f$ then either $(u,v)$ is already an edge of $\hat{G}$
or $u$ and $v$ coincide with $s$ and $t$, respectively. In either
case, $(u,v)$ is an edge of $G$. This proves (1). Further, since $f$
is a face of $\hat{G}$, every edge of this clique of $G$ induced by
$f$ lies outside or on the boundary of $f$. This proves (2). \qed
\end{proof}

We are now ready to prove the following.

\begin{lemma}\label{le:3n-6} If $G$ is a maximal almost-planar topological graph, then either
$\hat{G}$ is a maximal planar graph (that is, every face of
$\hat{G}$ has size 3); or every face of $\hat{G}$ has size 3, except
exactly one face $f_4$ which has the following properties: (i)$f_4$
has size 4; (ii)$f_4$ induces a clique in $G$; and (iii) both $s$ and $t$
are on $f_4$.
\end{lemma}
\begin{proof}
Let $G$ be a maximal almost-planar topological graph. Let $m$ denote
the number of edges of $G$. Since $\hat{G}$ is planar, $m \leq
3n-5$. Suppose that $\hat{G}$ is not a maximal planar graph, that
is, $m \leq 3n-6$.

First consider the case that $m \leq 3n - 7$, that is, $\hat{G}$ has
at most $3n - 8$ edges. Simple counting shows that $\hat{G}$ either
contains a face with more than four vertices or it has at least two
faces each having four vertices.

Suppose first that $\hat{G}$ has a face $f$ with $k \geq 5$
vertices. From Lemma~\ref{le:faces}(1), $\hat{G}_f$ has at least
$\frac{k(k-1)}{2} - 1$ edges. From Lemma~\ref{le:faces}(2),
$\hat{G}_f$ has at most $2k - 2$ edges. However, $\frac{k(k-1)}{2}
-1 > 2k -2$ if $k > 4$; thus $\hat{G}$ cannot have a face with at
least five vertices.

Suppose now that $\hat{G}$ has two faces $f_1$ and $f_2$ each having
four vertices. From Lemma~\ref{le:faces}, $G_{f_i}$ has 6 edges and
$\hat{G}_{f_i}$ has 5 edges, for both $i=1$ and $i=2$. Thus $s$ and
$t$ must be vertices of both $f_1$ and  $f_2$. Hence, $f_1$ and
$f_2$ are as in Fig.~\ref{fi:two-faces}; consider vertices $u$ and
$v$ in the figure: they must be adjacent in $\hat{G}$, but edge
$(u,v)$ would either split face $f_1$ or introduce a crossing. It
follows that $\hat{G}$ cannot have two faces with four vertices.
\begin{figure}[ht!]
  \centering
  \includegraphics[width=.5\columnwidth]{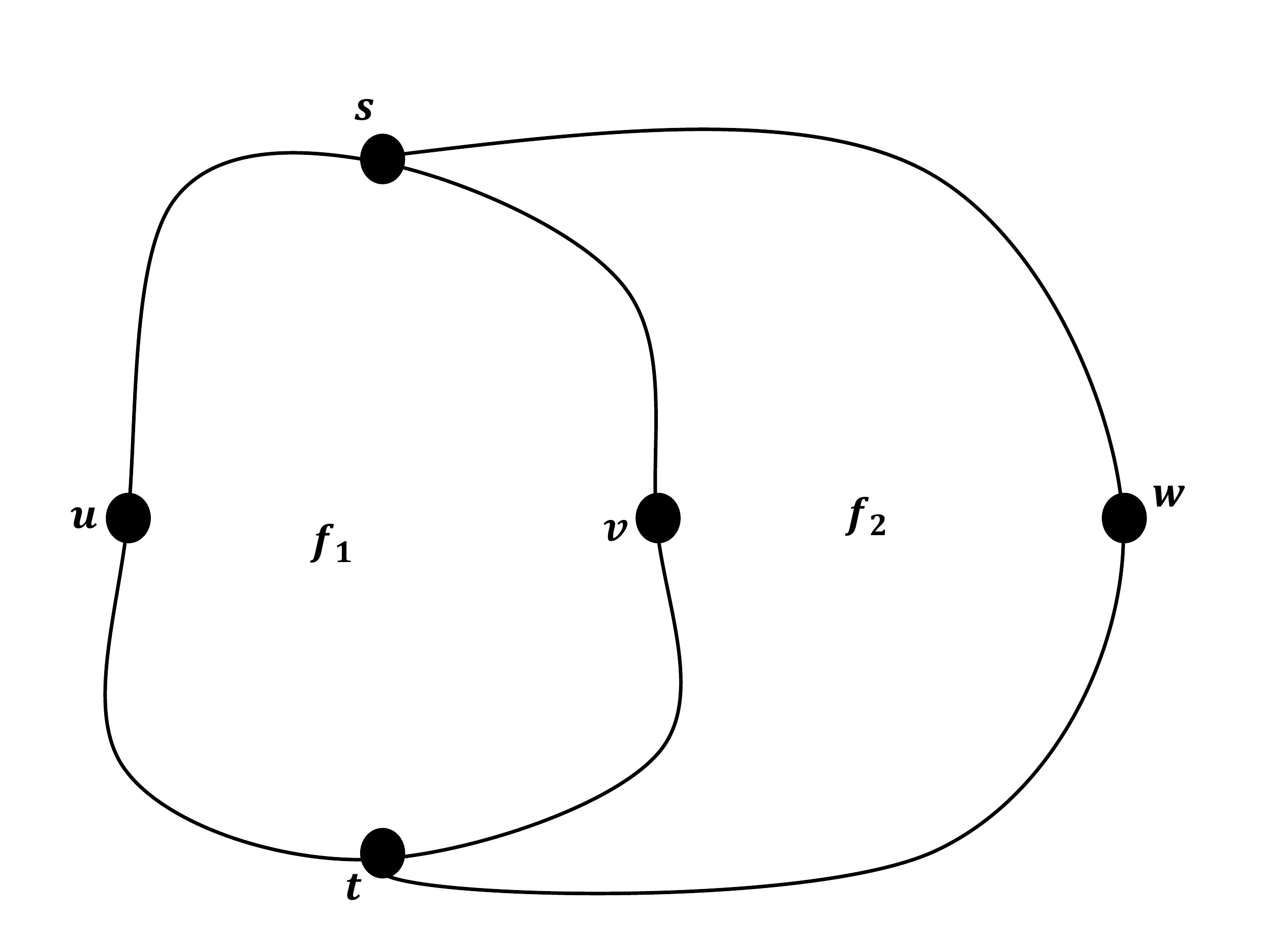}
  \caption{
  If $\hat{G}$ has two faces each having  four vertices, then $G$ is not a maximal almost-planar graph.}
  \label{fi:two-faces}
\end{figure}
Finally consider the case that $m=3n-6$. Counting reveals that every
face of $\hat{G}$ has size 3, except exactly one face $f_4$ has size
4. The other properties of $f_4$ follow from Lemma~\ref{le:faces}.
\end{proof}

\subsection{Choice of an external face}\label{ss:external-face}

%Let $G'$ be the $\mathbb{S}^2$-topological graph computed by
%Lemma~\ref{le:augmentation}. We want to identify a face $f_o$ of
%$G'$ such that if we project $G'$ on the plane with $f_0$ as the
%external face, then $G'$ becomes an $\mathbb{R}^2$-topological graph
%that satisfies a consistency condition on the faces of a suitable
%subgraph that we define in the next paragraph. We shall in fact
%prove in the next sections that if a maximal almost-planar graph has
%an $\mathbb{R}^2$-embedding that satisfies both the vertex
%consistency and this second condition, then the graph has an
%embedding preserving straight-line drawing in the plane.

The augmentation step results in a maximal almost-planar
$\mathbb{S}^2$-topological graph $G'$ in which every vertex is
consistent. Next, we want to identify a face $f_o$ of $G'$ such that
if we choose $f_0$ to be the external face, then $G'$ becomes an
$\mathbb{R}^2$-topological graph that has an embedding preserving
straight-line drawing in the plane. To identify such a face, we need
some further terminology.

 Let $G$ be an
almost-planar topological graph. Let $\hat{G}$ denote $G - (s,t)$.
We denote the set of left (resp. right) vertices of $G$ by $V_L$
(resp. $V_R$). We denote the subgraph of $\hat{G}$ induced by $V_L
\cup \{s,t \}$ (resp. $V_R \cup \{s,t \}$) by $\hat{G}_L$ (resp.
$\hat{G}_R$). The union of $\hat{G}_L$ and $\hat{G}_R$ is
$\hat{G}_{LR}$, and $G_{LR}$ denotes the topological subgraph of $G$
formed from $\hat{G}_{LR}$ by adding the edge $(s,t)$. Note that
$G_{LR}$ and $\hat{G}_{LR}$ are \emph{not} necessarily induced
subgraphs of $\hat{G}$. A face of $G_{LR}$ is {\em inconsistent} if
it contains a left vertex and a right vertex, and \emph{consistent}
otherwise. In fact we can prove that vertex consistency implies that
$G_{LR}$ has exactly one inconsistent face.
But first we prove some simple results about
$G_{LR}$ .

\begin{lemma}
\label{le:cycle} Let $G$ be an $\mathbb{S}^2$-topological graph in
which every vertex is consistent and let $C$ be a simple cycle in
$G_{LR}$ with at least one left vertex and at least one right
vertex. Then $C$ contains $s$ and $t$.
\end{lemma}
\begin{proof}
Note that $G_{LR}$ has no edge that joins a left vertex to a right
vertex. Thus all paths from a left vertex to a right vertex pass
through either $s$ or $t$, and thus $C$ contains both $s$ and $t$.
\qed
\end{proof}

\begin{lemma}
\label{le:outerplanar} If $G$ is an almost-planar topological graph
then $\hat{G}_{LR}$ is outerplanar.
\end{lemma}
\begin{proof}
The edge $(s,t)$ passes through a face $f$ of $\hat{G}_{LR}$, and
$s$ and $t$ are on that face. Further, every left vertex $\ell$ is
on $f$ because, in $G$, there is an edge incident to $\ell$ that
crosses $(s,t)$. Similarly every right vertex is on $f$. \qed
\end{proof}

We first prove that $G_{LR}$ cannot have two inconsistent faces and
then show that $G_{LR}$ always has one inconsistent face.

\begin{lemma}\label{le:at-most-one-inconsistent-face}
Let $G$ be an almost-planar $\mathbb{S}^2$-topological graph in
which every vertex is consistent. Then $G_{LR}$ has at most one
inconsistent face.
\end{lemma}
\begin{proof}
Suppose that $G_{LR}$ had two inconsistent faces, say $f_0$ and
$f_1$. Make a projection on the plane so that $f_0$ becomes the
external face of an  $\mathbb{R}^2$-embedding of $G_{LR}$. The
boundary of $f_0$ contains a simple cycle with at least one left
vertex and at least one right vertex; by Lemma~\ref{le:cycle},
vertices $s$ and $t$ are vertices along the boundary of $f_0$. Also,
there exists a simple path on the boundary of $f_0$ that starts at
$s$, ends at $t$, and is such that any other vertex of the path is a
left vertex; we call such a path the {\em left path} of $f_0$ and
denote it as $\Pi_L$. Similarly, the {\em right path} $\Pi_R$ of
$f_0$ is the simple path along the boundary of $f_0$ whose
endvertices are $s$ and $t$ and such that any internal vertex is a
right vertex of $f_0$. Observe that edge $(s,t)$ is inside face
$f_0$ since $f_0$ is chosen as the external face in the
$\mathbb{R}^2$-embedding of $G$.

The {\em left cycle} of $G$ is the simple cycle consisting of $\Pi_L
\cup (s,t)$; the {\em right cycle} of $G$ is the simple cycle
consisting of $\Pi_R \cup (s,t)$. Every vertex that does not belong
to $f_0$ is either strictly inside the left cycle or strictly inside
the right cycle. Note that any vertex $v$ of $G_{LR}$ inside the
left cycle must be a left vertex. Namely, $v$ is the endvertex of an
edge $e$ that crosses edge $(s,t)$. Since $G$ is almost-planar, edge
$e$ cannot cross any other edge of $G$ except $(s,t)$. Therefore
edge $e$ crosses $(s,t)$ without intersecting the boundary of $f_0$,
which implies that $v$ is a left vertex. Similarly, every vertex of
$G_{LR}$ inside the right cycle is a right vertex.

Consider now face $f_1$; $f_1$ must be entirely inside either the
left cycle or entirely inside the right cycle. However, $f_1$ is
inconsistent and therefore it has both left vertices and right
vertices, which contradicts the fact that all vertices inside the
left cycle are left vertices and all vertices inside the right cycle
are right vertices.\qed
\end{proof}

\begin{lemma}\label{le:at-least-one-inconsistent-face}
Suppose that $G$ is a maximal almost-planar
$\mathbb{S}^2$-topological graph in which every vertex is
consistent. Then $G_{LR}$ has at least one inconsistent face, and
this face is a simple cycle.
\end{lemma}
\begin{proof}
Let $q_L$ be a shortest path in $\hat{G}_L$ from $s$ to $t$, and let
$q_R$ be a shortest path in $\hat{G}_R$ from $t$ to $s$. The
existence of such paths is guaranteed by maximality. Further, since
$(s,t)$ is not in either $\hat{G}_L$ or $\hat{G}_R$, $q_L$ contains
at least one left vertex and $q_R$ contains at least one right
vertex. From Lemma~\ref{le:outerplanar}, the cycle $C$ formed by
concatenating $q_L$ and $q_R$ forms a face $f$ of $\hat{G}_{LR}$.

If $f$ is also a face of $G_{LR}$, then the Lemma follows. Otherwise
$(s,t)$ must ``split'' $f$ in $\hat{G}_{LR}$. In this case, every
left and every right vertex must lie on $f$; it follows that
$\hat{G}_{LR}$ is exactly $f$. The Lemma follows, since $(s,t)$ can
be on only one side of the cycle $C$. \qed
\end{proof}

Lemma~\ref{le:at-most-one-inconsistent-face} and
Lemma~\ref{le:at-least-one-inconsistent-face} imply the following.

\begin{lemma}
\label{le:oneFaceIsInconsistent} Let $G$ be an $\mathbb{S}^2$-topological graph in
which every vertex is consistent. Then $G_{LR}$ has exactly one inconsistent face.
\end{lemma}

We now proceed as follows.
Let
$G$ be an $\mathbb{S}^2$-topological graph in which every vertex is
consistent and let $G'$ be a maximal almost-planar $\mathbb{S}^2$-embedding
preserving augmentation of $G$ constructed by using
Lemma~\ref{le:augmentation}. We project $G$ on the plane such that
the only inconsistent face of $G_{LR}$ is its external face. The
following lemma is a consequence of the discussion above and of
Lemma~\ref{le:oneFaceIsInconsistent}.

\begin{lemma}\label{le:properties-of-G'}
Let $G$ be a maximal almost-planar $\mathbb{S}^2$-topological graph
in which every vertex is consistent. There exists an
$\mathbb{R}^2$-topological graph $G'$ that preserves the
$\mathbb{S}^2$-embedding of $G$ and such that: (i) every internal
face of $\hat{G'}$ consists of three vertices (i.e. it is a
triangle); (ii) every internal face of $G'_{LR}$ is consistent.
\end{lemma}

Examples of an almost-planar $\mathbb{R}^2$-topological graph $G$
and of its subgraphs $\hat{G_L}$, $\hat{G_R}$, and $G_{LR}$ are
given in Fig.~\ref{fi:exampleBigAll0} and
Fig.~\ref{fi:exampleBigAll}.

\begin{figure}[h!]
  \centering
  \includegraphics[width=.99\columnwidth]{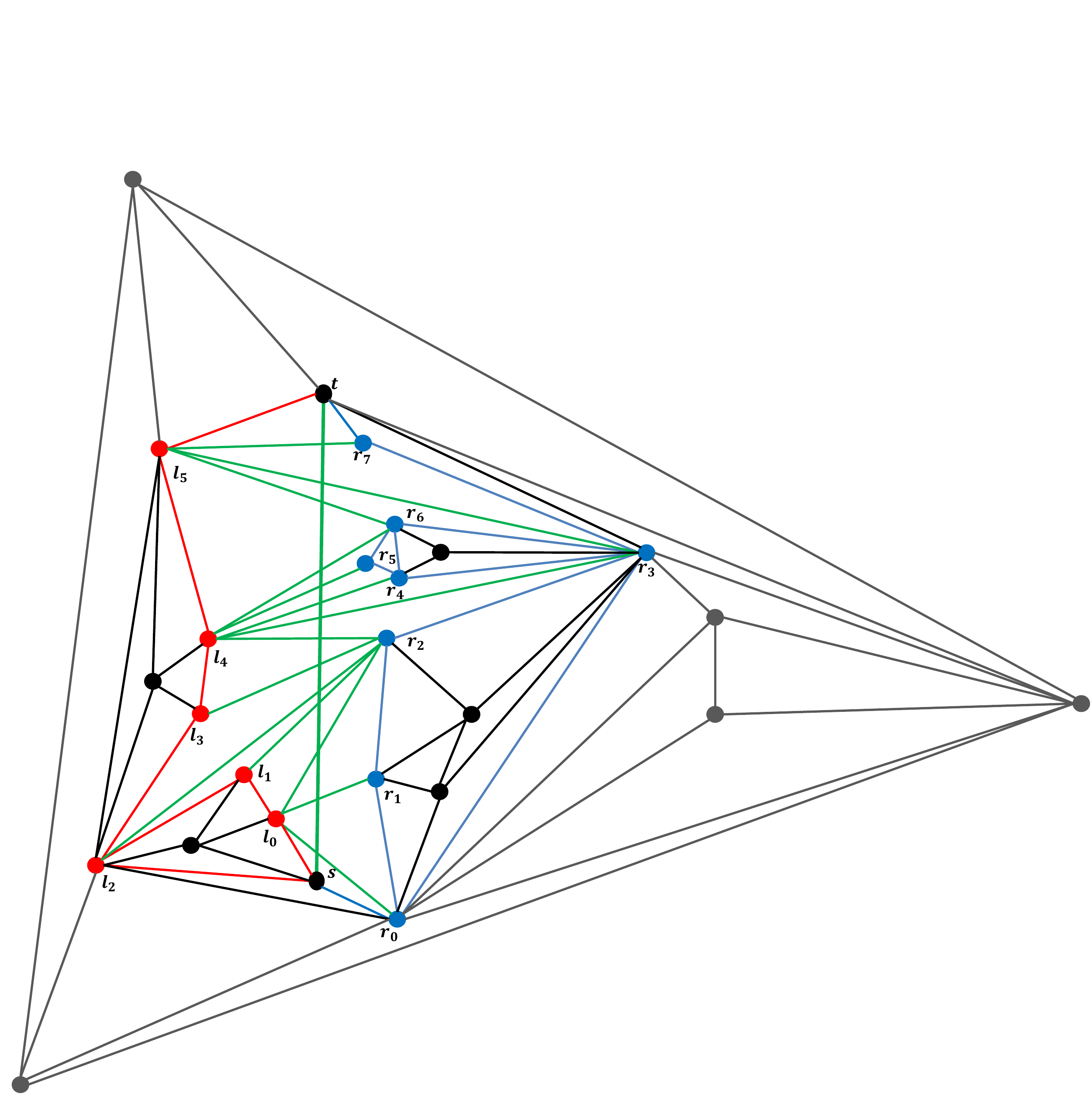}
  \caption{An almost-planar $\mathbb{R}^2$-topological graph $G$.
  Left vertices are red and labelled $l_i$, $i=0, \ldots, 5$; right vertices are blue and labelled $r_i$, $i=0, \ldots, 7$.}
  \label{fi:exampleBigAll0}
\end{figure}

\begin{figure}[h!]
  \centering
  \includegraphics[width=.99\columnwidth]{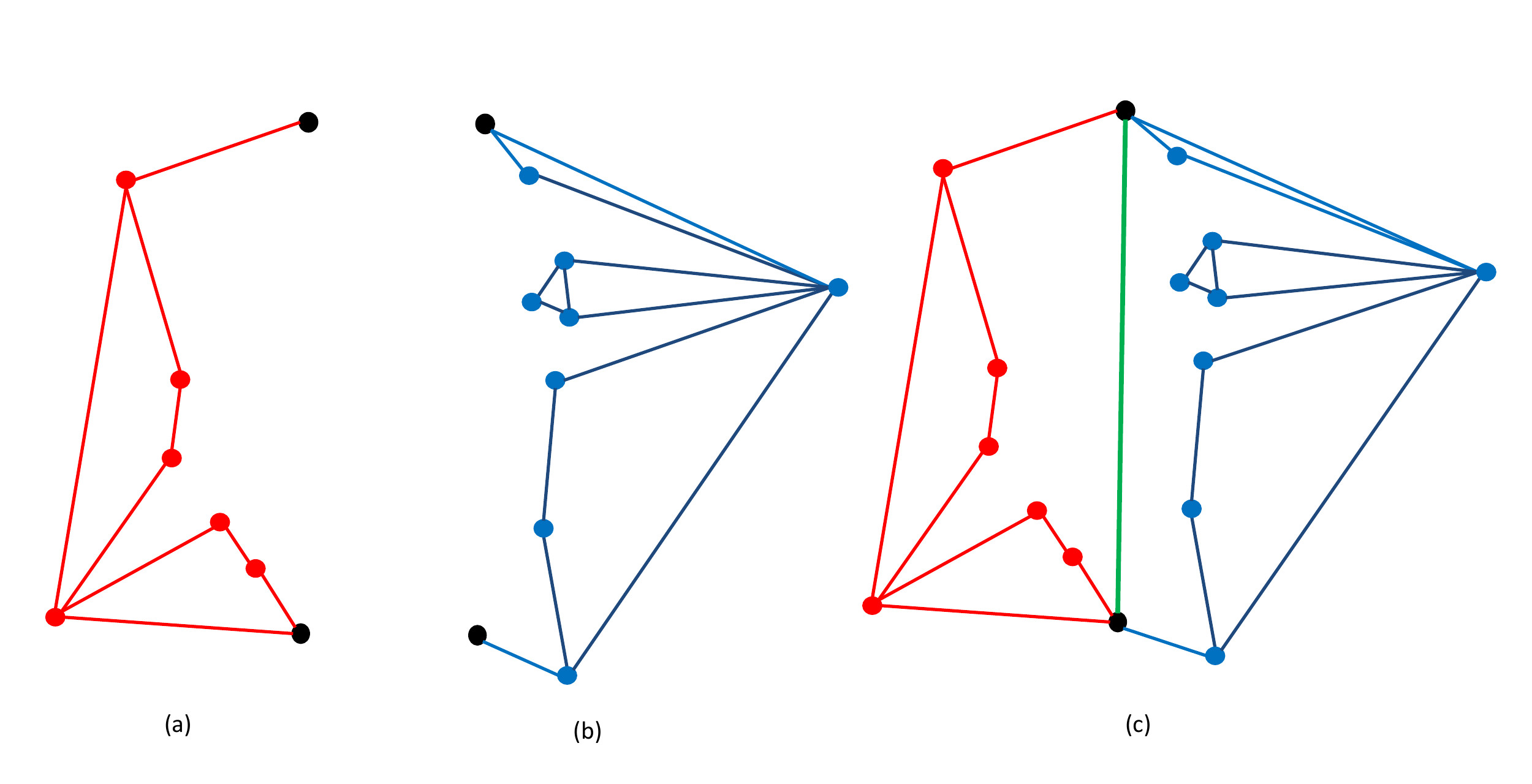}
  \caption{Subgraphs of the graph $G$ in Fig.~\ref{fi:exampleBigAll0}: (a) $\hat{G}_L$, (b) $\hat{G}_R$, (c) $G_{LR}$.}
  \label{fi:exampleBigAll}
\end{figure}

%This can be done, for example, by looking for the inconsistent face
%$f_i$ of $G_{LR}$; this face divides the spherical embedding of $G$
%into two parts one of which includes edge $(s,t)$. We choose a face
%$f_o$ of $G$ which is on the opposite side of $(s,t)$ with respect
%to $f_i$ and project on the plane such that $f_o$ becomes the
%outside face. It is easy to see that this procedure can be executed
%in linear time.

\subsection{Splitting the augmented graph}\label{ss:divide}

For the remainder of Section~\ref{se:s2},
we assume that $G$ is a maximal almost-planar $\mathbb{R}^2$-topological graph;
that is, that the augmentation
and choice of an outer face have been done.
Next we divide $G$ into the ``inner graph''  and the
``outer graph''.

Denote the induced subgraph of $\hat{G}$ on the vertex set $V_L \cup
V_R \cup \{s, t \}$ by $\hat{G}^+_{LR}$. Note that $\hat{G}_{LR}$ is
a subgraph of $\hat{G}^+_{LR}$, but these graphs may not be the
same;  in particular, $\hat{G}^+_{LR}$ may have edges with a left
endpoint and a right endpoint that do not cross $(s,t)$; such an
edge is called a \emph{cap} edge. Although $\hat{G}$ is internally
triangulated by Lemma~\ref{le:properties-of-G'}, $\hat{G}^+_{LR}$
may have {\em non-triangular} inner faces.
Fig.~\ref{fi:GplusLRGin}(a) shows $\hat{G}^+_{LR}$, where $G$ is the
graph in Fig.~\ref{fi:exampleBigAll0}. The non-triangular inner
faces in Fig.~\ref{fi:GplusLRGin}(a) arise, for example, from cycles
of left vertices in Fig.~\ref{fi:exampleBigAll0} that have vertices
that are neither left nor right in their interior.
Examples are in Fig.~\ref{fi:GplusLRGin}(a).

\begin{figure}[ht!]
  \centering
  \includegraphics[width=.9\columnwidth]{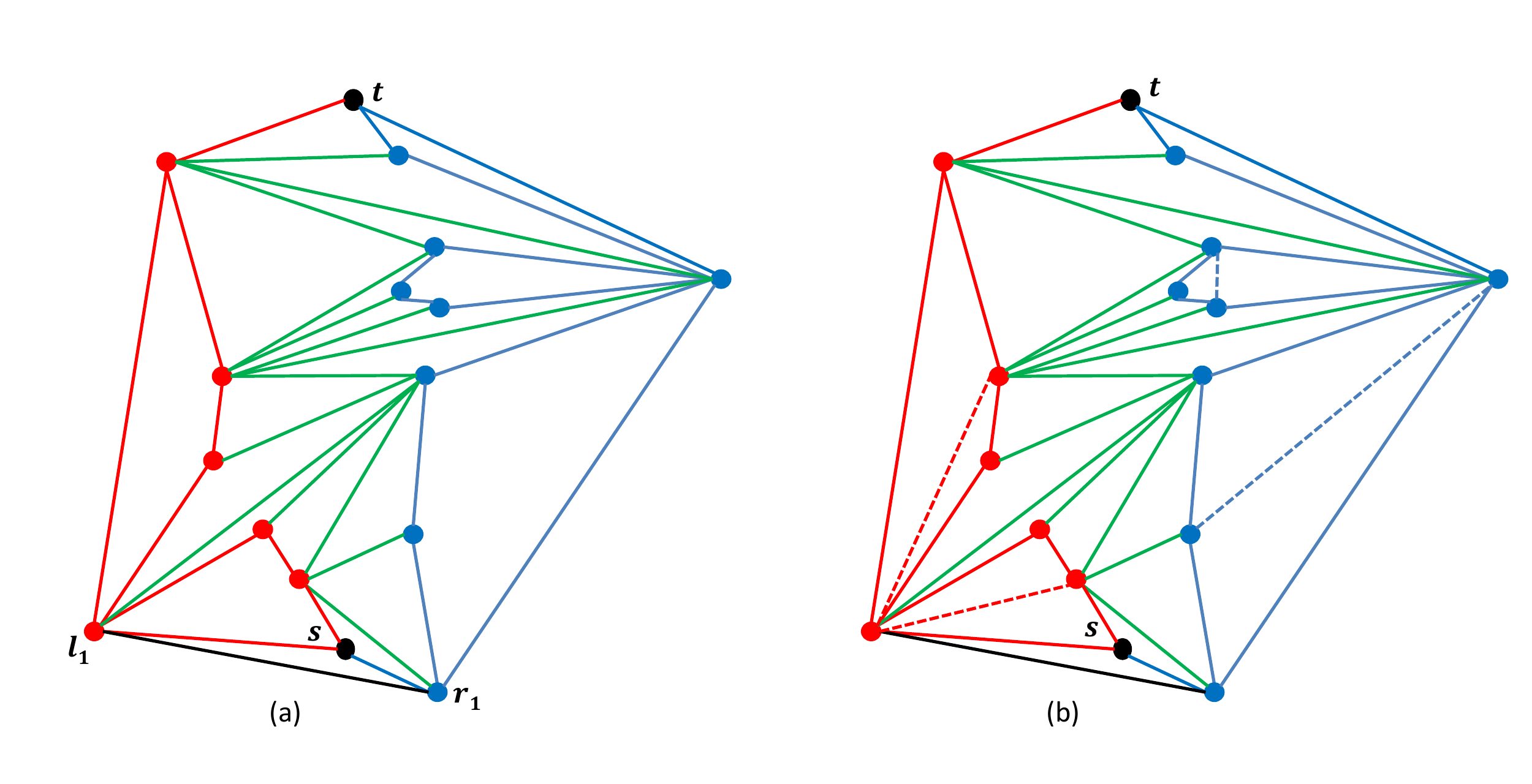}
  \caption{
  (a) $\hat{G}^+_{LR}$, where $G$ is in Fig.~\ref{fi:exampleBigAll0}; here $(\ell_1, r_1)$ is a \emph{cap} edge.
  (b) The \emph{inner graph} $\hat{G}_{in}$.}
  \label{fi:GplusLRGin}
\end{figure}

To following lemma states that the external face of $\hat{G}^+_{LR}$
is a simple cycle. It can be proved with the same method as used in
the proof of Lemma~\ref{le:at-least-one-inconsistent-face}. We can
show that the concatenation $C$ of the shortest paths $q_L$ and
$q_R$ is the boundary of the external face of $\hat{G}_{LR}$; thus
the external face of $\hat{G}_{LR}$ is a simple cycle.
 Since $\hat{G}^+_{LR}$ is an induced subgraph with the same
vertex set as $\hat{G}_{LR}$, it follows that the external face of
$\hat{G}^+_{LR}$ is a simple cycle.

\begin{lemma}
\label{le:simpleCycleOutside}
If $G$ is a maximal almost-planar $\mathbb{R}^2$-topological graph such
that every internal face of $G_{LR}$ is consistent, then the external face
of $\hat{G}^+_{LR}$ is a simple cycle.
\end{lemma}

We call the external face of $\hat{G}^+_{LR}$ the \emph{separating
cycle} of the graph $G$. The topological subgraph consisting of the
separating cycle as well as all vertices and edges that lie outside
the separating cycle is the \emph{outer graph} $G_{out}$.
Fig.\ref{fi:exampleBigOuter} is an example of outer graph.

\begin{figure}[htb!]
  \centering
  \includegraphics[width=.8\columnwidth]{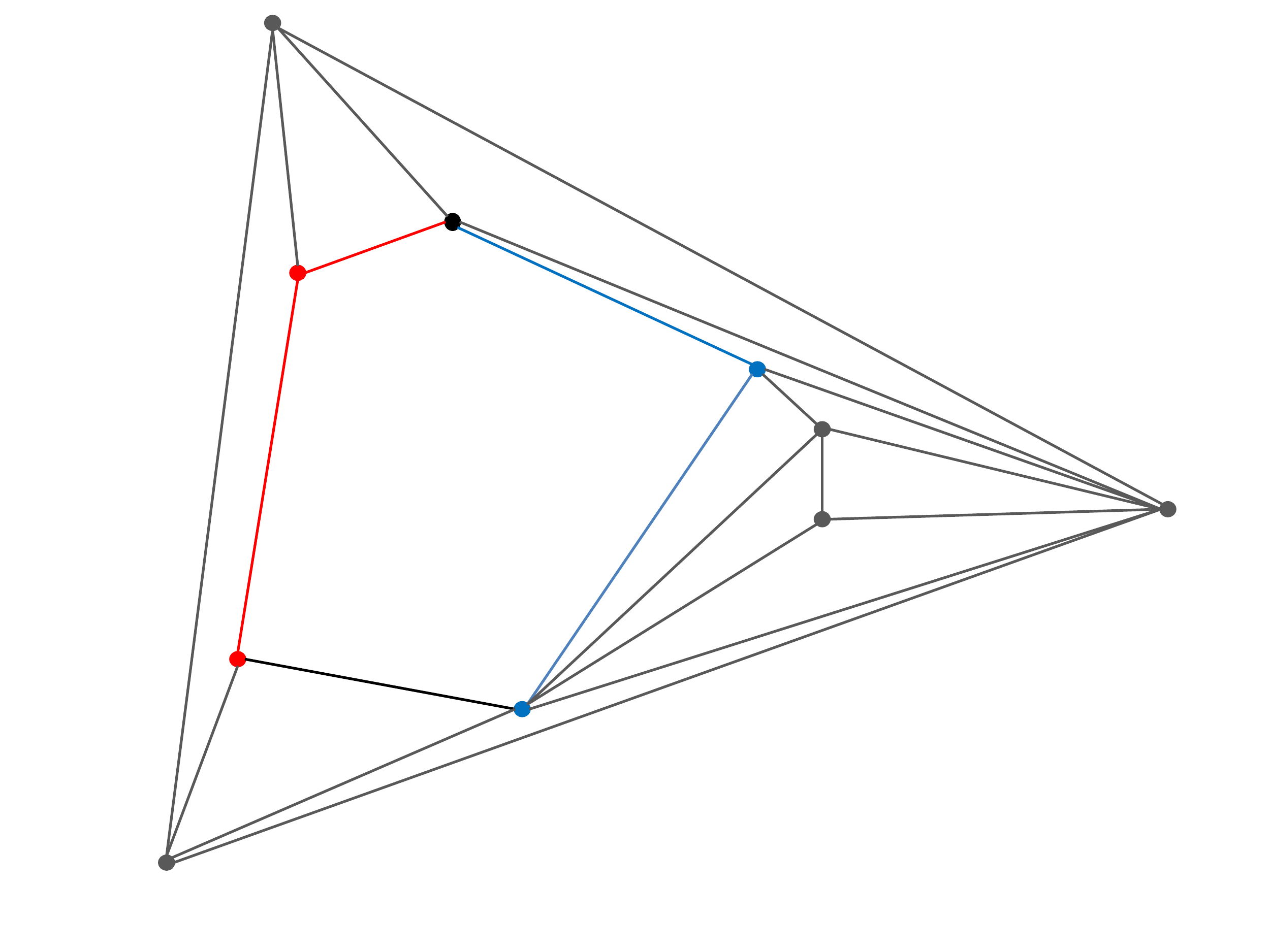}
  \caption{The outer graph for the graph in Fig.~\ref{fi:exampleBigAll0}(a).}
  \label{fi:exampleBigOuter}
\end{figure}

 The {\em inner graph} consists of
$\hat{G}^+_{LR}$ with the addition of some dummy edges. Namely, for
every face $f$ of $\hat{G}^+_{LR}$ that is not a
triangle, we perform a \emph{fan triangulation}; that is, we choose
a vertex $u$ of $f$ with degree 2 in $f$, and add dummy edges incident with $u$ to
triangulate $f$. The graph formed by fan triangulating every
non-triangular internal face of $\hat{G}^+_{LR}$ is the \emph{inner
graph} $\hat{G}_{in}$.

Note that the vertices of $G$ that are neither left vertices nor
right vertices and that are inside the separating cycle belong to
neither the inner nor the outer graph. At the end of next section we
show how to reinsert these vertices and their incident edges into
the drawing.

%The following lemma can be proved by making traversing  a standard
%data structure for the representation of an almost-planar graph.
%
%\begin{lemma}\label{le:split-lemma}
%Let $G$ be a maximal almost-planar $\mathbb{R}^2$-topological graph
%with $n$ vertices such that  every internal face of $\hat{G}$
%consists of three vertices and every internal face of $G_{LR}$ is
%consistent. Graphs $G_{out}$ and $\hat{G}_{in}$ can be computed in
%$O(n)$ time.
%\end{lemma}

%\section{Computing a straight-line drawing of a maximal $\mathbb{R}^2$-topological graph }\label{se:drawing}
%
%In this section we first show how to draw the outergraph, then how
%to draw the inner graph, and finally how to ``patch'' the remaining
%part of the graph. We also give a characterization of the maximal
%almost-planar  $\mathbb{R}^2$-topological graph that have a
%straight-line drawing.

\subsection{Drawing the outer graph}
\label{ss:outerGraph}

Since $G$ is maximal almost-planar, by using Lemma~\ref{le:3n-6} we
can show that $G_{out}$ is triconnected as long as the separating
cycle has no chord. But since $G_{in}$ contains the subgraph of
$\hat{G}$ induced by the separating cycle, every chord on the
separating cycle is in $G_{in}$ and not in $G_{out}$. Thus $G_{out}$ is
triconnected. We use the linear-time convex drawing algorithm of
Chiba et al.~\cite{CYN} to draw $G_{out}$ such that every face in
the drawing is a convex polygon. This drawing of the outer graph has
a convex polygonal drawing of the separating cycle, which we shall
call the \emph{separating polygon}. In the next section we show how
to draw the inner graph such that its outside face (i.e. the
separating cycle) is the separating polygon.

\subsection{Drawing the inner graph}
\label{ss:innerGraph}

%In this Section we describe a linear time algorithm to draw the inner graph.
%\begin{lemma}
%\label{le:inner}
%The plane graph $G_{in}$ can be drawn such that the separating
%cycle is drawn as a given convex polygon, in linear time.
%\end{lemma}
%One can draw the inner graph using a kind of weighted barycentre algorithm; see Appendix X.
%However, the barycentre approach does not yield a linear-time algorithm.
%In this Section we describe a more complex but linear-time approach.
The overall approach for drawing the inner graph is described as
follows. For each edge $e$ of the separating cycle, we define a
``side graph'' $S_e$; intuitively, $S_e$ consists on vertices and edges that are ``close'' to $e$.
There may be two special side graphs, that
contain cap edges (that is, edges that join a left vertex and a
right vertex but do not cross $(s,t)$); these side graphs are ``cap
graphs''. Each side graph has a block-cutvertex tree $T_e$. We root
$T_e$ at the block (biconnected component) that contains the edge
$e$. The algorithm first draws the root block for each side graph,
then proceeds from the root to the leaves of these trees, drawing
the blocks one by one. Cap graphs are drawn with a different
algorithm from that used for other side graphs.

Each non-root block $B$ in $T_e$ with parent cutvertex $c$ is associated with circular arc
$\gamma(B)$, and two regions, called a ``safe wedge'' $\omega(B)$
and a ``safe moon'' $\mu(c)$; these are defined precisely below. We
draw all the vertices of $B$ and its descendants in $\mu(c)$,
with all vertices except $c$ lying on $\gamma(B)$
inside $\mu(c) \cap \omega(B)$. Every edge with exactly one endpoint
in $B$ and its descendants lies inside $\omega(B)$.

First the root blocks are drawn, and then the algorithm proceeds by
repeating the following steps until every vertex of every side graph
is drawn. (1) Choose a ``safe block'' $B$ from the child
blocks of drawn vertices; (2) Compute the ``safe moon'' $\mu(c)$,
the ``safe wedge'' $\omega(B)$, and the circular arc $\gamma(B)$;
(3) Draw each vertex of $B$ except $c$ on
$\gamma(B)$.

\smallskip\noindent{\bf Side graphs and cap graphs.} To define ``side
graphs'' and ``cap graphs'', we need to first define a certain
closed walk in the inner graph. Denote the edges that cross $(s,t)$
by $e_0, e_1, \dots , e_{p-1}$, ordered from $s$ to $t$ by their
crossing points along $(s,t)$.
%as in Fig.~\ref{fi:walk0}.
Suppose that $e_i = (\ell_i , r_i)$ for $0 \leq i \leq p-1$, where
$\ell_i$ is a left vertex and $r_i$ is a right vertex. Note that
cyclic list $(s, \ell_0, \ell_1 , \ldots , \ell_{p-1}, t, r_{p-1},
\ldots , r_1 , r_0 )$ may contain repeated vertices.

%\begin{figure}[ht!]
%  \centering
%  \includegraphics[width=.5\columnwidth]{walk0.pdf}
%  \caption{Here $W$ is the closed walk $(s,e,d,c,b,a,t,g,h,i,h,j,k)$.}
%  \label{fi:walk0}
%\end{figure}

 Now let $W$ be the sublist of $(s, \ell_0, \ell_1 ,
\ldots , \ell_{p-1}, t, r_{p-1}, \ldots , r_1 , r_0 )$ obtained by
replacing each contiguous subsequence of the same vertex by a single
occurrence of that vertex. Note that $W$ may contain repeated
vertices, but these repeats are not contiguous. Namely, $W$ is a
closed spanning walk of $\hat{G}_{LR}$. An
example of this walk $W$ is in Fig.~\ref{fi:walk}.

\begin{figure}[h!]
  \centering
  \includegraphics[width=0.8\columnwidth]{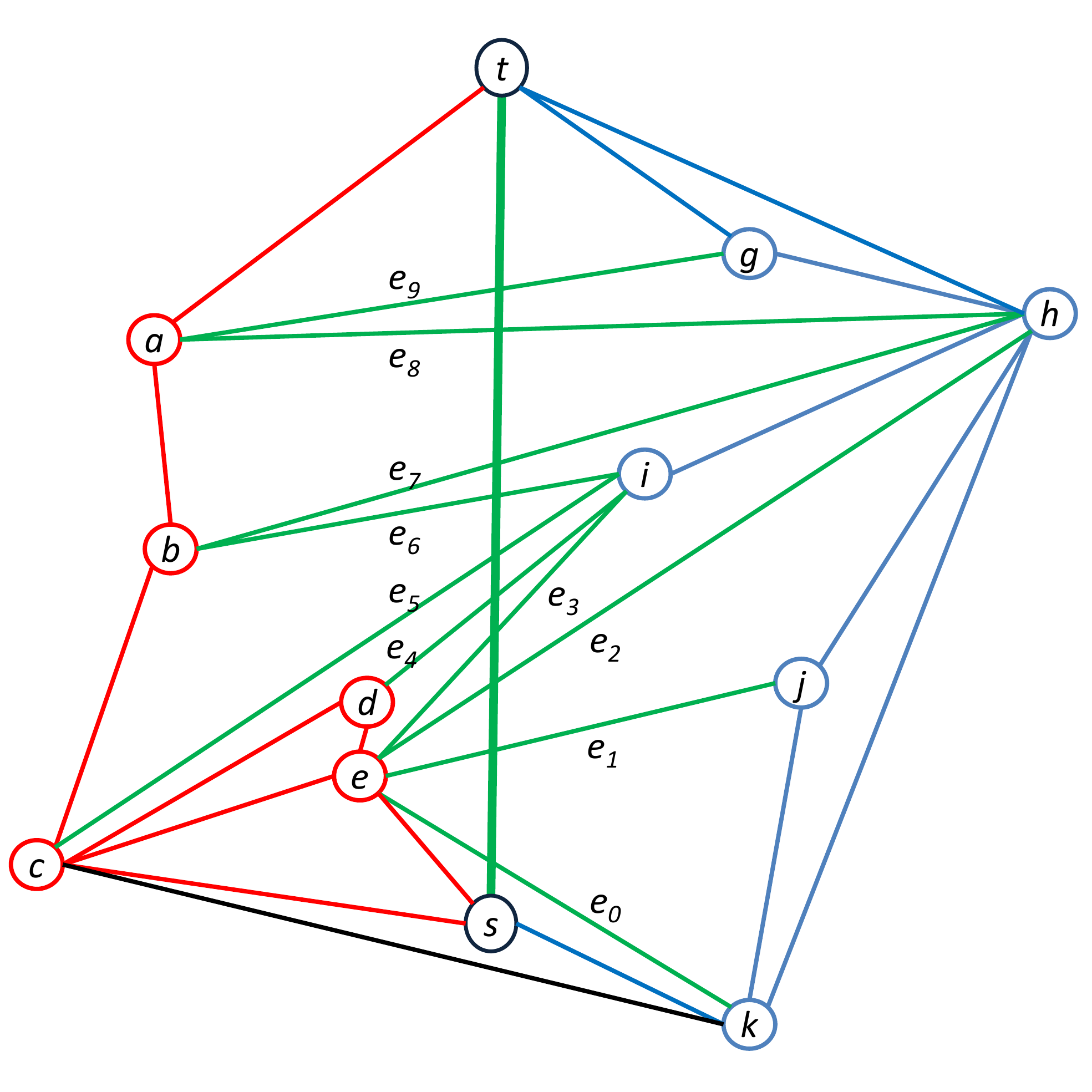}
  \caption{$W$ is the closed walk $(s,e,d,c,b,a,t,g,h,i,h,j,k)$.}
  \label{fi:walk}
\end{figure}

Now let $e = (u, v)$ be an edge of the separating cycle, with $u$
before $v$ in clockwise order around the separating cycle. Note that
both $u$ and $v$ are elements of the closed walk $W$. Suppose that
the clockwise sequence of vertices in $W$ between $u$ and $v$ is $(u
= u_1 , u_2 , \ldots , u_k = v)$. If $u$ occurs more than once in
$W$, then we choose $u_1$ to be the first occurrence of $u$ in
clockwise order after $s$; similarly choose $u_k$. The \emph{side
graph} $S_e$ is the induced subgraph of $G$ on $\{u_1 , u_2 , \ldots
, u_k \}$.

If $S_e$ contains both left and right vertices then it is a
\emph{cap graph}. Note that a cap graph contains either $s$ or $t$;
one can show that $s$ and $t$ are not in the same cap graph.
Examples of side graphs, including cap graphs, are in Fig.~\ref{fi:side}.

\begin{figure}[h!]
  \centering
  \includegraphics[width=0.8\columnwidth]{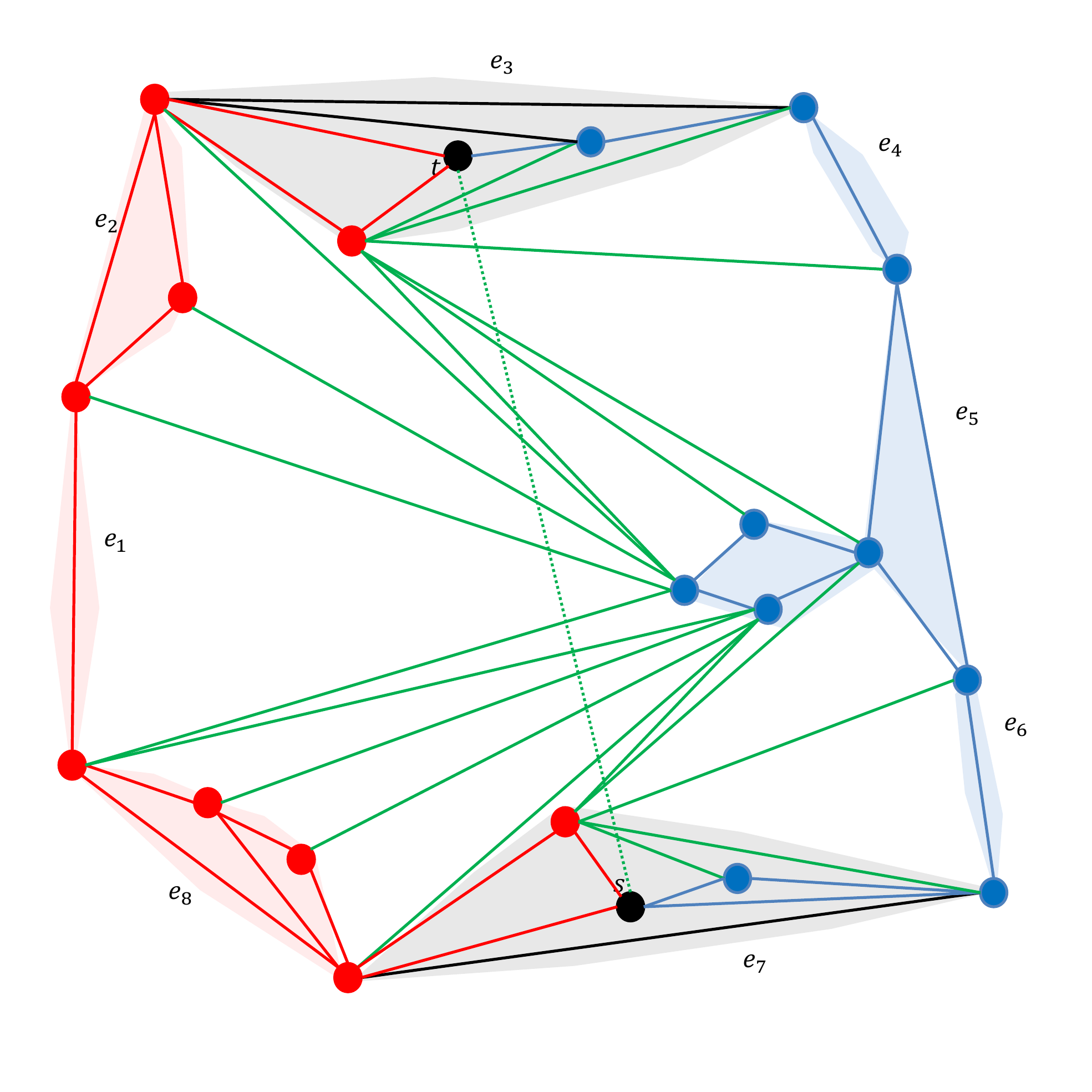}
  \caption{An inner graph with eight side graphs (shaded).
  This graph has two cap graphs $S_{e_3}$ and $S_{e_7}$.}
  \label{fi:side}
\end{figure}

%\begin{figure}[ht!]
%  \centering
%  \includegraphics[width=.5\columnwidth]{side.pdf}
%  \caption{Side graphs}.
%  \label{fi:side}
%\end{figure}
%

\smallskip\noindent{\bf Drawing the root blocks of side graphs.}
Next we show how to draw the root block $B^*_e$ of the side graph
$S_e$. The edge $e$ is drawn as a side $\lambda_e$ of the separating
polygon. We define a circular arc $\gamma(B^*_e)$ through the
endpoints of $\lambda_e$, with radius chosen such that the maximum
distance from $\lambda_e$ to $\gamma(B^*_e)$ is $\epsilon_1$. We
will show how to choose $\epsilon_1$ later; for the moment, we
assume that $\epsilon_1$ is very small in comparison to the length
of the smallest edge of the separating polygon. The convex region
bounded by $\lambda_e$ and $\gamma(B^*_e)$ is called the
\emph{pillow} of $e$.

Suppose that $B^*_e$ of $S_e$ has $a+1$ vertices, which occur in
clockwise order on the closed walk $W$ as $w_0 , w_1, \ldots, w_a$.
Since $B^*_e$ is biconnected, this sequence is a Hamilton path of
$S_e$. We compute $a+1$ equally spaced points $\alpha(w_0),
\alpha(w_1), \ldots , \alpha(w_a)$ on $\lambda_e$ as in
Fig.~\ref{fi:pillows}(a). Let $\zeta(w_i)$ denote the line through
$\alpha(w_i)$ orthogonal to $\lambda_e$, as in
Fig.~\ref{fi:pillows}(a).

\begin{figure}[htb!]
  \centering
  \includegraphics[width=1\columnwidth]{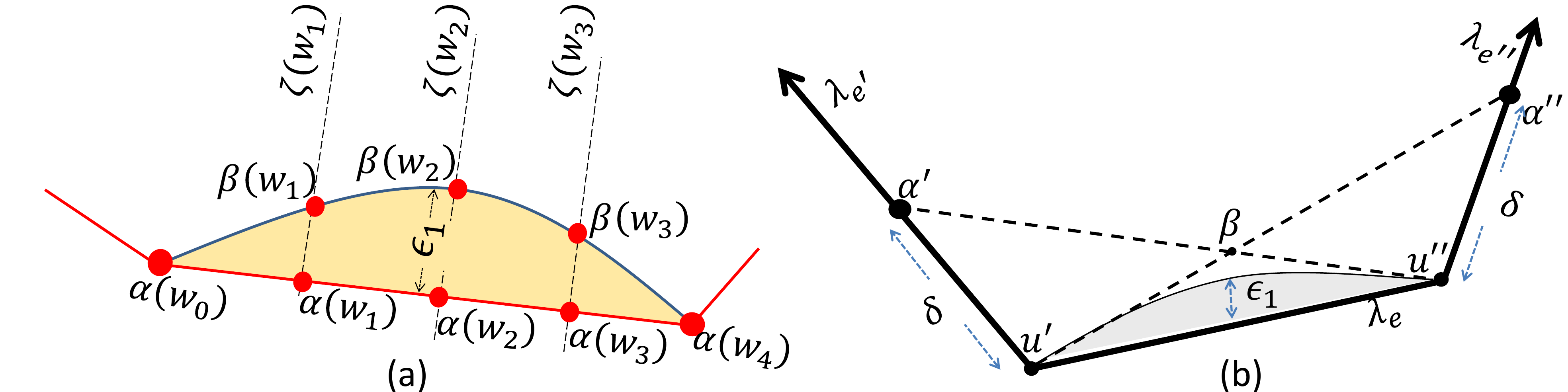}
  \caption{(a) A pillow. (b) Defining $\epsilon_1$.
  %Here there is a cap graph containing $s$. For this cap graph, some vertices are on the boundary of the pillow, others are strictly inside the pillow. Vertices of other side graphs are all on the boundaries of their respective pillows. {\bf note: it may help to also draw the monotone hamilt. path of the cap graphs}
  }
  \label{fi:pillows}
\end{figure}

%\begin{figure}[ht!]
%  \centering
%  \includegraphics[width=.5\columnwidth]{pillow.pdf}
%  \caption{A pillow.
%  %(b) Defining $\epsilon_1$.
%  %Here there is a cap graph containing $s$. For this cap graph, some vertices are on the boundary of the pillow, others are strictly inside the pillow. Vertices of other side graphs are all on the boundaries of their respective pillows. {\bf note: it may help to also draw the monotone hamilt. path of the cap graphs}
%  }
%  \label{fi:pillow}
%\end{figure}

If $S_e$ is not a cap graph, then we simply place vertex $w_i$ in
$B^*_e$ at the point $\beta(w_i)$ where $\zeta(w_i)$ intersects the
circular arc $\gamma(B^*_e)$ ($0 \leq i \leq a$). Note that the
edges of $S_e$ (which are chords on the Hamilton path $(w_0 , w_1,
\ldots, w_a )$) lie within the pillow of $e$.

If $S_e$ is a cap graph, then we place vertex $w_i$ on the line
$\zeta(w_i)$, but not necessarily at $\beta(w_i)$. First we define
an acyclic directed graph as follows.
%Let $B^{*}_e$ be the induced
%subgraph of $G$ on $\{ w_0 , w_1, \ldots, w_a \}$. Essentially
%$B^{*}_e$ is formed from $B^*_e$ by adding those edges with both
%endpoints in $B^*_e$ that cross $(s,t)$.
We direct edges along the Hamilton path $(w_0 , w_1, \ldots, w_a )$
from $w_0$ to $w_a$, and direct other edges so that the result is a
directed acyclic graph $\overrightarrow{B^{*}_e}$ with a source at
$w_0$ and a sink at $w_a$. Note that $\overrightarrow{B^{*}_e}$ is a
\emph{leveled planar graph} with one vertex on each
level~\cite{DBLP:journals/jgaa/JungerL02}. One can use the algorithm
in~\cite{EFLN} to draw $\overrightarrow{B^{*}_e}$ so that there are
no edge crossings, vertex $w_i$ lies on the line $\zeta(w_i)$, and
the external face is a given polygon. We choose the external face to
be the convex hull of $\lambda_e$ and the points $\beta(w_i)$, $0 \leq i
\leq a$. Note that the vertices $w_0 , w_1, \ldots, w_a$ are in
monotonic order in the direction of the edge $(w_0, w_k)$. The
general picture after the drawing of the root blocks is illustrated
in Fig.~\ref{fi:pillowsGeneral}.

\begin{figure}[ht!]
\centering
 \includegraphics[width=.6\columnwidth]{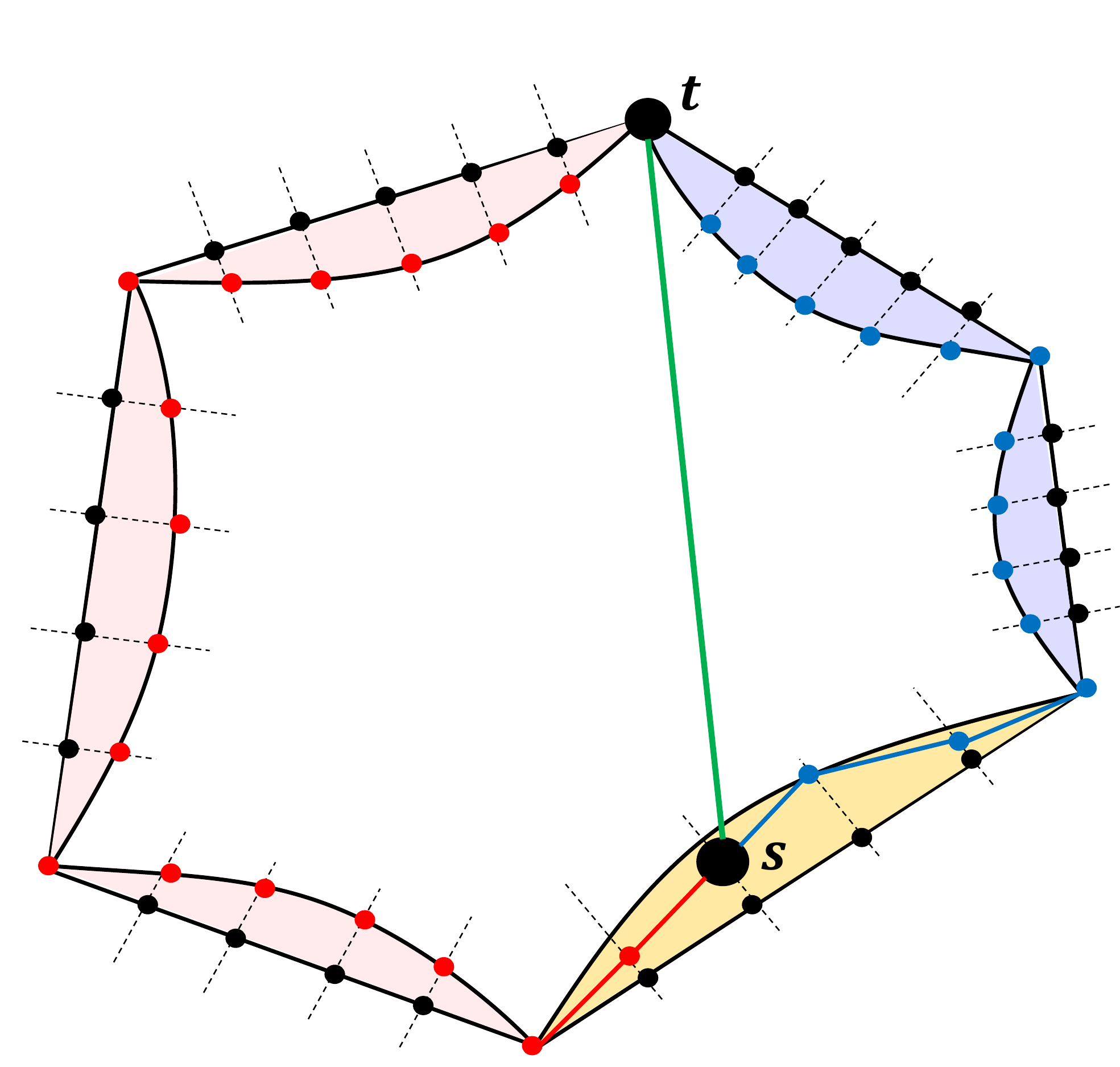}
  \caption{The general picture with pillows.}
  \label{fi:pillowsGeneral}
\end{figure}

%\begin{wrapfigure}[18]{r}[0pt]{.4\columnwidth}
%\centering
% \includegraphics[width=.4\columnwidth]{pillowsGeneral.pdf}
%  \caption{The general picture with pillows.}
%  \label{fi:pillowsGeneral}
%\end{wrapfigure}

Next we show how to choose $\epsilon_1$. Let $\delta$ denote $d/n$,
where $d$ is the minimum length of a side of the separating polygon,
and $n$ is the number of vertices in the graph. Suppose that
$\lambda_{e'}$, $\lambda_{e}$, and $\lambda_{e''}$ are three
consecutive sides of the separating polygon, as in
Fig.~\ref{fi:pillows}(b); we show how to choose $\epsilon_1$ for
the edge $e$. Suppose that the endpoints of $e$ are $u'$ and $u''$,
and $\alpha'$ and $\alpha''$ are points on $\lambda_{e'}$ and
$\lambda_{e''}$ distant $\delta$ from $u'$ and $u''$ respectively.
Suppose that the line from $u'$ to $\alpha''$ meets the line from
$u''$ to $\alpha'$ at $\beta$. Convexity ensures that $\beta$ is
inside the separating polygon, and thus $(u' , \beta , u'' )$ forms a
triangle inside the separating polygon. We choose $\epsilon_1$ so
that the circular arc $\gamma(B^*_e)$ through $u'$ and $u''$ lies
inside this triangle (meeting the triangle only on the line segment
$\lambda_{e}$). The reason for this choice of $\epsilon_1$ is to
ensure that all vertices in $B^*_e$ are so close to the side
$\lambda_e$ of the separating polygon that it is impossible for an
edge between different pillows to intersect with pillows other than
those at its endpoints.

\smallskip\noindent{\bf Safe blocks.}
To describe the algorithm for drawing the non-root blocks, we need
some terminology. Suppose that $c$ is a cutvertex in the side graph
$S_e$, and $B = (V_B, E_B)$ is a child block of $c$. Suppose that
$c$ is a left vertex. In the clockwise order of edges in $G$
around $c$, there is an edge $e_1 \not \in E_B$, followed by a
number of edges in $E_B$, followed by an edge $e_2 \not \in E_B$, as
illustrated in Fig.~\ref{fi:boundingEdges0}(a). We say that $e_1$
and $e_2$ are the \emph{bounding edges} of $B$. Note that a bounding
edge either crosses $(s,t)$, or has $s$ or $t$ as an endpoint.

\begin{figure}[ht!]
  \centering
  \includegraphics[width=1\columnwidth]{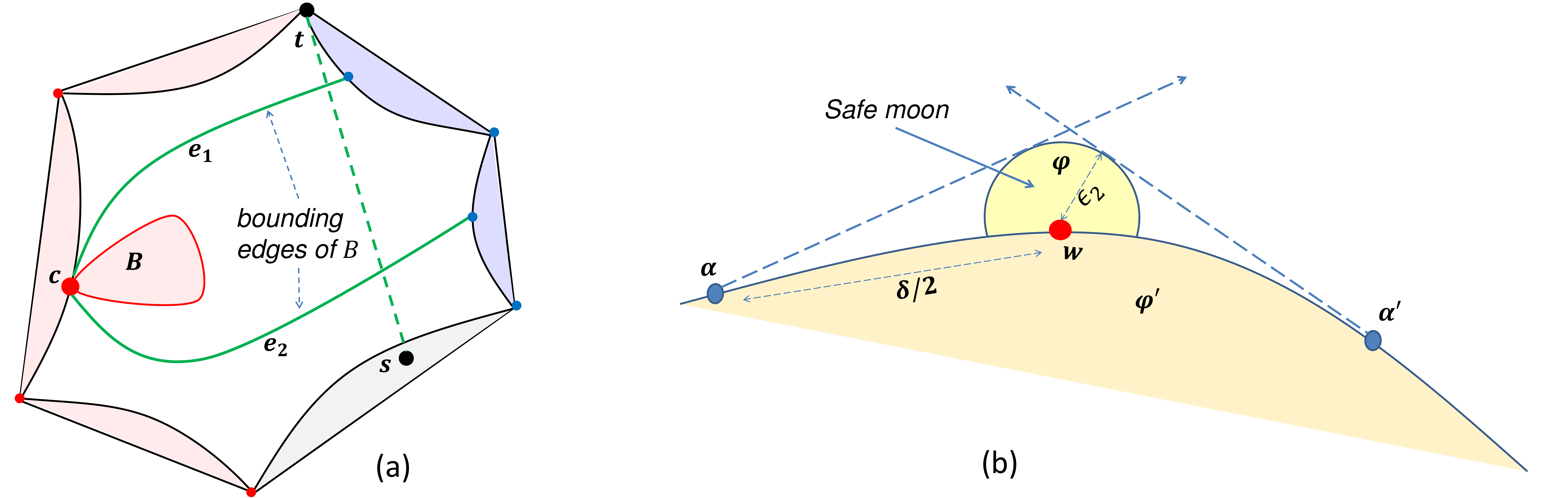}
  \caption{(a) Bounding edges of a block. (b)The safe moon $\mu(u)$ at $u$}.
  \label{fi:boundingEdges0}
\end{figure}

 At any stage of the drawing algorithm, a block may be
{\em safe} or {\em unsafe}. A block $B$ is \emph{safe} if the
following properties hold: (i) The parent cutvertex $c$ (that is,
the parent of $B$ in the block-cutvertex tree) has been drawn, and
the other vertices in $B$ are not drawn; (ii) Suppose that the
boundary edges of $B$ are $e_1 = (c, u_1)$ and $e_2= (c, u_2)$; let
$u'_1$ and $u'_2$ be the vertices which are the least already-drawn
ancestors of $u_1$ and $u_2$ respectively in their respective
block-cutvertex trees. Then we require that $u'_1 \not = u'_2$.

\begin{lemma}\label{le:safeBlockExists}
If there is an undrawn vertex, then there is a safe block.
\end{lemma}
\begin{proof}
Let $B$ be a block with parent cutvertex $c$ and boundary edges $(c,
u_1)$ and $(c, u_2)$; let $u'_1$ and $u'_2$ be the least
already-drawn ancestors of $u_1$ and $u_2$ respectively in their
respective block-cutvertex trees. If $B$ is not a safe block, we
have that $u'_1 = u'_2$.

Two cases are possible: either $u_1$ and $u_2$ belong to the same
block with ancestor cutvertex $u'_1 = u'_2$ or they are in different
blocks.

\begin{figure}[ht!]
  \centering
  \includegraphics[width=.5\columnwidth]{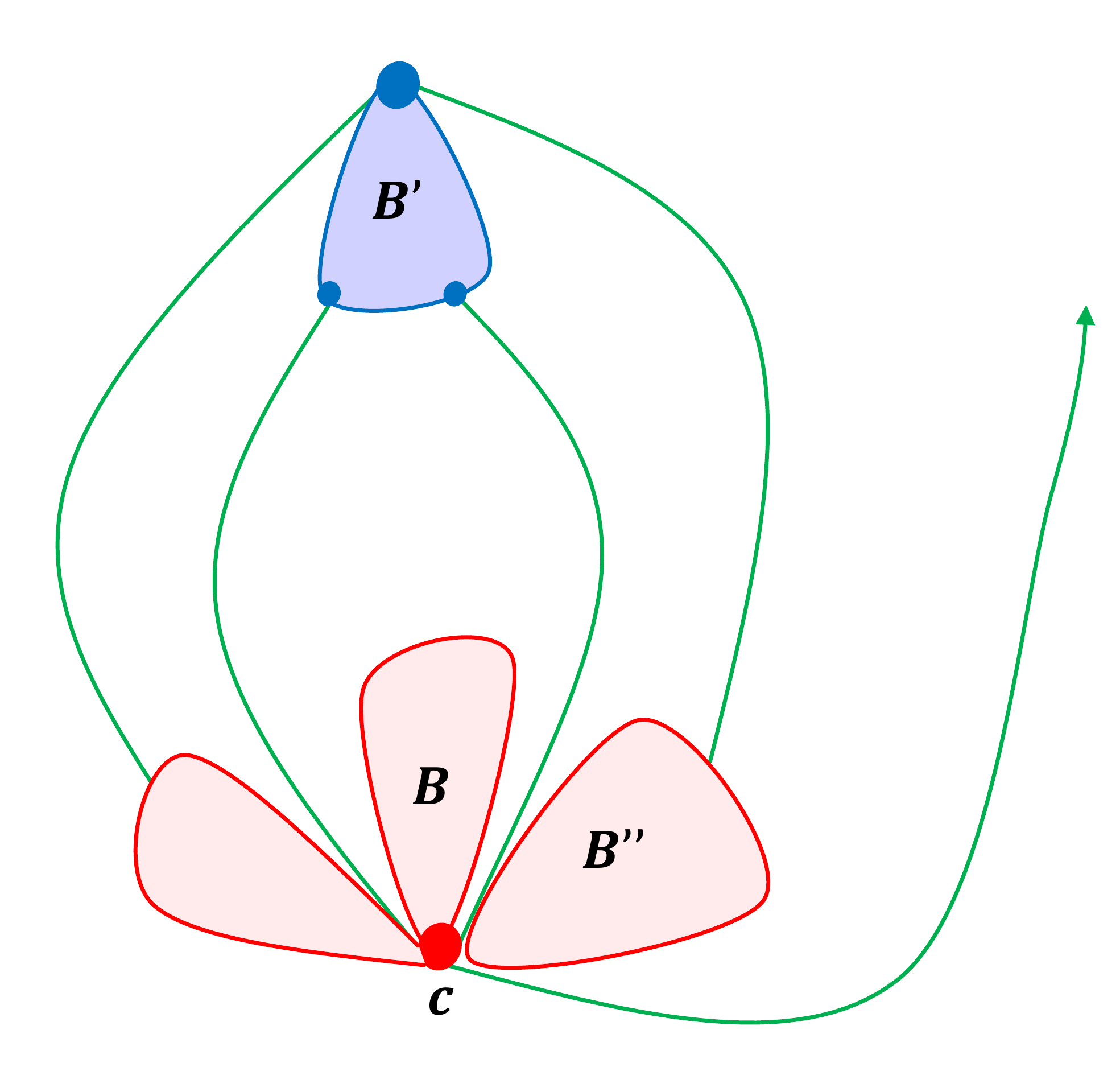}
  \caption{For the proof of Lemma~\ref{le:safeBlockExists}: either $B'$ or $B''$ is safe.}
  \label{fi:safeBlock}
\end{figure}

Consider the first case and refer to Fig.~\ref{fi:safeBlock}. Let
$B'$ be the block of both $u_1$ and $u_2$. Since the inner graph is
plane, neither boundary edge of $B'$ is incident to a vertex of $B$.
Either $B'$ is safe and we are done, or the boundary edges of $B'$
are also incident to undrawn vertices of two blocks that have $c$ as
their parent cutvertex. Again, either one of these blocks is safe
and we are done, or their boundary edges must be incident to undrawn
vertices of blocks whose least already-drawn ancestor is $u'_1$. By
repeating this argument, we find a block $B''$ having either $u'_1$
or $c$  as its parent cutvertex and such that at least one boundary
edge of $B''$ is incident to a vertex whose least already-drawn
ancestor differs from both $u_1'$ and $c$. Hence $B''$ is a safe
block.

With similar reasoning, the existence of a safe block can be
proved when $u_1$ and $u_2$ belong to different blocks with parent
cutvertex $u'_1 = u'_2$.\qed
\end{proof}

\smallskip\noindent{\bf Safe moon.}
Suppose that $w$ is a parent cutvertex for a safe block $B$; for the
moment we assume that $w$ is not on the separating cycle. Suppose
that the parent block of $w$ is $B'$; then $w$ has been drawn on the
circular arc $\gamma(B')$. Denote the circular disc defined by
$\gamma(B')$ by $\phi'$. Let $\phi$ be a circular disc of radius
$\epsilon_2$ with centre at $w$. We show how to choose $\epsilon_2$
later; for the moment we can assume that $\epsilon_2$ is very small
in comparison to the radius of $\gamma(B')$. The \emph{safe moon}
$\mu(w)$ for $w$ is the interior of $\phi - \phi'$; see
Fig.\ref{fi:boundingEdges0}(b).

%\begin{figure}[ht!]
%  \centering
%  \includegraphics[width=.5\columnwidth]{safeMoon.pdf}
%  \caption{The safe moon $\mu(u)$ at $u$}
%  \label{fi:safeMoon}
%\end{figure}

 Now we show how to choose $\epsilon_2$. Again let
$\delta$ denote $d/n$, where $d$ denotes the minimum length of a
side of the separating polygon, and $n$ is the number of vertices in
the graph. Now consider two points $\alpha$ and $\alpha'$ at
distance $\frac{\delta}{2}$ from $u$. We choose $\epsilon_2$ small
enough that: (i) $\mu(w)$ at $u$ does not intersect the tangents to
$\gamma(B')$ at $\alpha$ and $\alpha'$; (ii) $\mu(w)$ does not
intersect the line through $s$ and $t$. Small adjustments to this
choice of $\mu(w)$ are required for the cases where $w$ is on the
separating cycle, and where $w$ is an endpoint of $\gamma(B')$.

A consequence of the definition of safe moon is the following: Let
$w_1$ and $w_2$ be vertices on the circular arcs $\gamma(B_1)$ and
$\gamma(B_2)$ for two blocks $B_1$ and $B_2$ that have been drawn.
Let $\alpha_1$ be a point in $\mu(w_1)$ and $\alpha_2$ be a point in
$\mu(w_2)$; the line segment between $\alpha_1$ and $\alpha_2$ does
not intersect any safe moon other than $\mu(w_1)$ and $\mu(w_2)$.

\smallskip\noindent{\bf Safe wedges.}
Suppose that the boundary edges of a non-root block $B$ are $e_1 =
(c, u_1)$ and $e_2= (c, u_2)$; let $u'_1$ and $u'_2$ be the vertices
which are the least drawn ancestors of $u_1$ and $u_2$ respectively
in their respective block-cutvertex trees. Since $B$ is safe, $u'_1
\not = u'_2$. For each point $\alpha_1$ (resp. $\alpha_2$) in
$\mu (u'_1)$ (resp. $\mu (u'_2)$), consider the wedge $\omega
(\alpha_1, \alpha_2)$ formed by the rays from $c$ through $\alpha_1$
and $\alpha_2$. The \emph{safe wedge} $\omega(B)$ of $B$ is the
intersection of all such wedges $\omega (\alpha_1, \alpha_2)$ with
the safe moon of $c$. This is illustrated in
Fig.~\ref{fi:circularArcGamma}(a).

\begin{figure}[htb]
  \centering
  \includegraphics[width=.8\columnwidth]{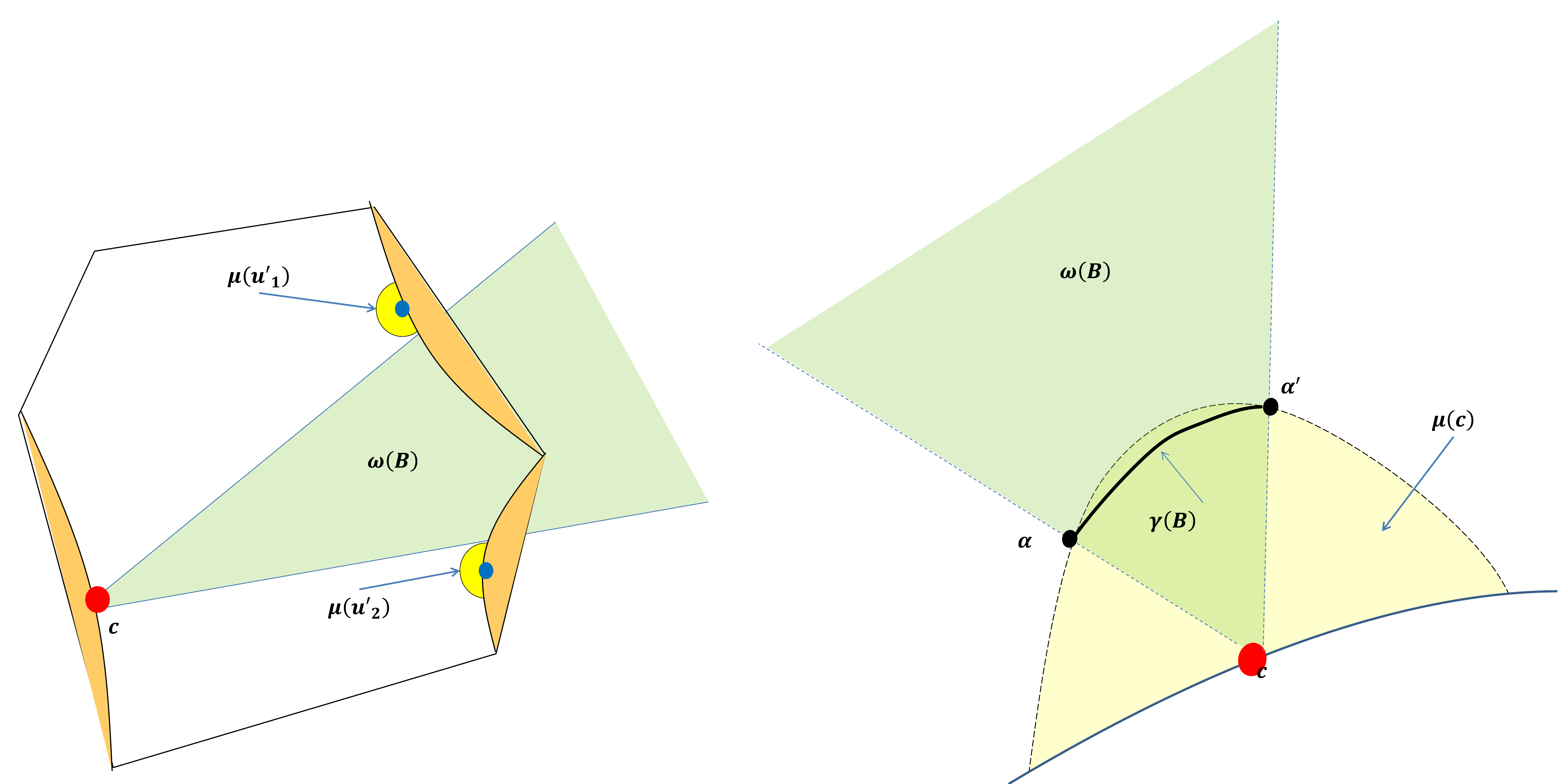}
  \caption{(a) A safe wedge. (b) The circular arc $\gamma(B)$.}
  \label{fi:circularArcGamma}
\end{figure}

%\begin{figure}[ht!]
%  \centering
%  \includegraphics[width=.5\columnwidth]{circularArc.pdf}
%  \caption{A safe wedge.}
%  \label{fi:circularArc}
%\end{figure}

\smallskip\noindent{\bf The circular arc $\gamma(B)$.}
Suppose that $B$ is a non-root block. We give a location to each
vertex in $B$ except the parent cutvertex $c$ (which is already
drawn). These vertices are drawn on a circular arc $\gamma(B)$,
defined as follows. Suppose that the boundaries of $\mu(c)$ and
$\omega(B)$ intersect at points $\alpha$ and $\alpha'$ as shown in
Fig~\ref{fi:circularArcGamma}(b). Then $\gamma(B)$ is a circular arc
that passes through $\alpha$ and $\alpha'$. The radius of
$\gamma(B)$ is chosen so that it lies inside $\mu(c)$, and it is
distant at most $\epsilon_1$ from the straight line between $\alpha$
and $\alpha'$. Here $\epsilon_1$ is chosen in exactly the same way
as for the root block.

%\begin{figure}[ht!]
%  \centering
%  \includegraphics[width=.5\columnwidth]{gamma.pdf}
%  \caption{A safe wedge. (b) The circular arc $\gamma(B)$.}
%  \label{fi:gamma}
%\end{figure}

\smallskip\noindent{\bf Putting it all together.}
In the construction of the inner graph in
subsection~\ref{ss:divide}, all vertices that are neither left nor
right are removed, and the resulting non-triangular faces are
fan-triangulated. These vertices can be drawn as follows. Each
fan-triangulated face, after removal of the dummy edges, is
\emph{star-shaped}. The non-aligned vertices (neither left nor right) that
came from this face form a triangulation inside the face. Thus we
can use the linear-time algorithm of Hong and Nagamochi~\cite{HN08}
to construct a straight-line drawing replacing the non-aligned vertices.
This concludes the proof of sufficiency of Theorem~\ref{th:s2}.

\section{Concluding Remarks}\label{se:wrap-up}

Assuming the real RAM model of computation, it can be proved that
all algorithmic steps presented in the previous section can be
executed in $O(n)$ time, where $n$ is the number of vertices of $G$.

\begin{theorem}\label{th:s2-drawing}
Let $G$ be an almost-planar $\mathbb{S}^2$-topological graph with
$n$ vertices such that every vertex of $G$ is consistent. There
exists an $O(n)$ time algorithm that computes an
$\mathbb{S}^2$-embedding preserving straight-line
 drawing of $G$.
\end{theorem}

The real RAM model of computation allows for exponentially bad
resolution in the drawing. The next theorem shows that such
exponentially bad resolution is inevitable in the worst case. The
construction of the family of almost-planar graphs for
Theorem~\ref{th:exponential} is based on a family of upward planar
digraphs first described by Di Battista et al.~\cite{DTT}. The graph
$G_1$ is illustrated in Fig.~\ref{fi:exponentialV2}(a), and for
$k>1$, the graph $G_k$ is formed from $G_{k-1}$ as illustrated in
Fig.~\ref{fi:exponentialV2}(b).
\begin{figure}[ht!]
  \centering
  \includegraphics[width=.9\columnwidth]{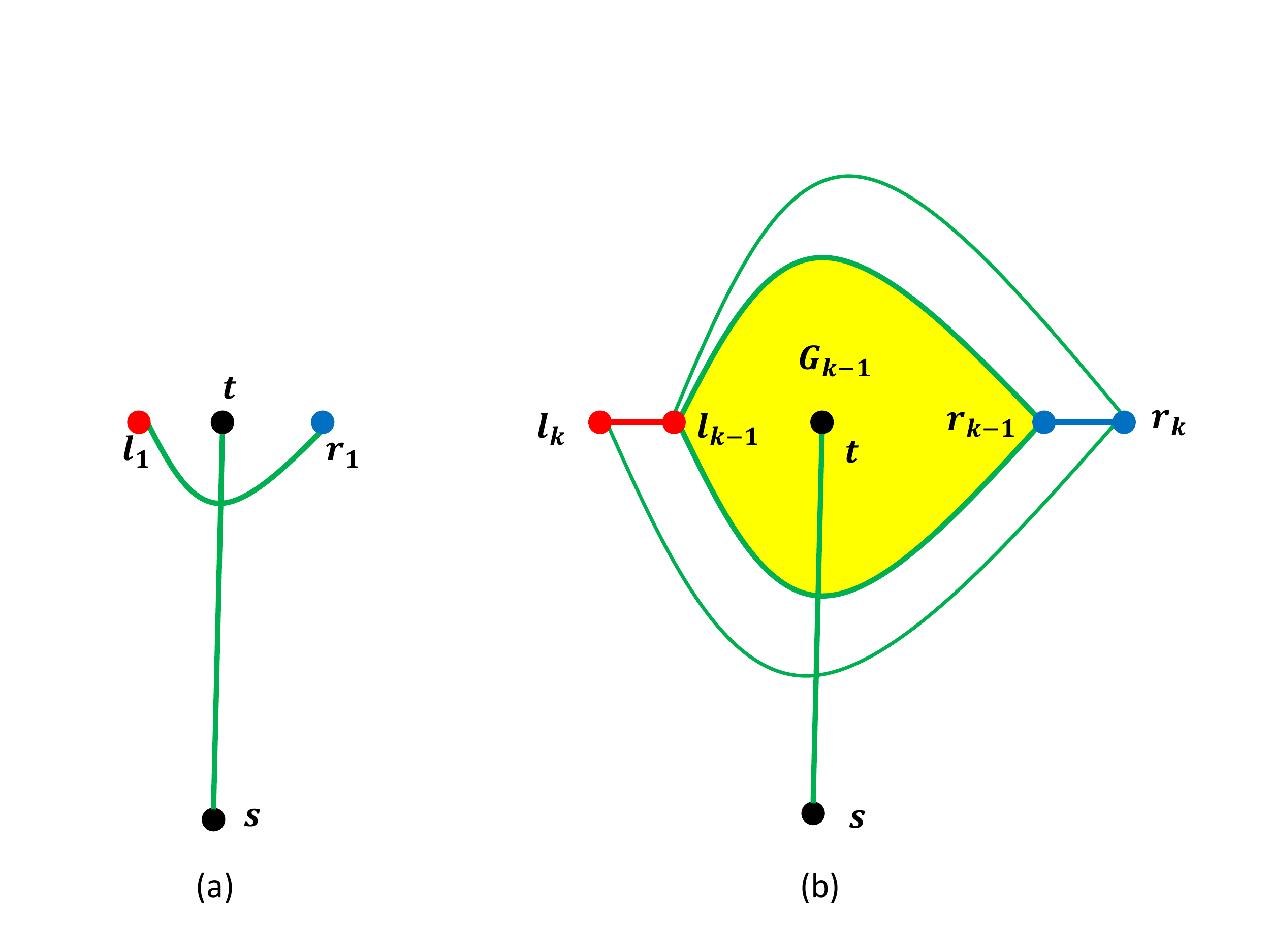}
  \caption{A family of graphs that require exponential area: (a) $G_1$; (b) Creating $G_k$ from $G_{k-1}$.}
  \label{fi:exponentialV2}
\end{figure}
The graph $G_5$ is illustrated in Fig.~\ref{fi:exponentialExample}.
\begin{figure}[ht!]
  \centering
  \includegraphics[width=.9\columnwidth]{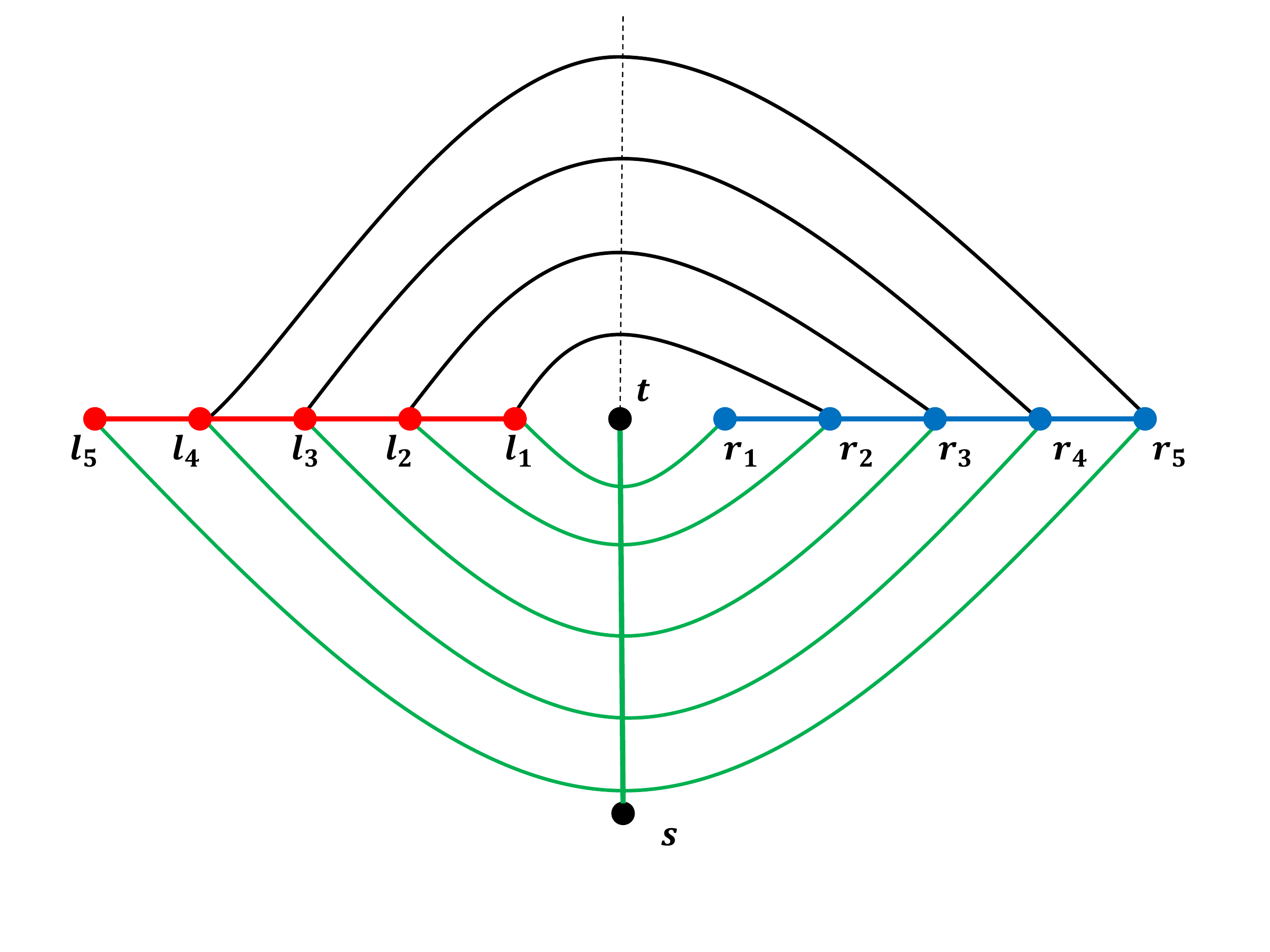}
  \caption{$G_5$.}
  \label{fi:exponentialExample}
\end{figure}

Note that $G_k$ consists of a path $(\ell_k , \ell_{k-1}, \ldots,
\ell_1, r_1, r_2, \ldots, r_k)$ and a number of chords on that path.
Some chords cross $(s,t)$ and some do not. The chords form a kind of
``spiral'' path $(r_1, \ell_1, r_2, \ell_2, \ldots , r_{k-1} ,
\ell_k, r_k)$, alternating between edges that cross $(s,t)$ and
edges that do not.
%For technical reasons we add
%edges $(\ell_k , s)$, and $(s, r_k)$, as in
%Fig.~\ref{fi:exponential}(b), to form a graph $G_k$.
%Denote the graph formed from $G_k$ by deleting the edge $(s,t)$ by $\hat{G}_k$.

It is easy to see (from Theorem~\ref{th:s2}) that $G_k$ has a
straight-line drawing.
%For an example, Fig.~\ref{fi:exponentialStraight} shows $G_3$, albeit with poor resolution.
%\begin{figure}[ht!]
%  \centering
%  \includegraphics[width=.9\columnwidth]{exponentialStraight.pdf}
%  \caption{A straight-line drawing of $G_3$ with bad resolution.}
%  \label{fi:exponentialStraight}
%\end{figure}

\begin{theorem}
\label{th:exponential} For each $k \geq 1$, there is an
almost-planar $\mathbb{S}^2$-topological graph $G_k$ with $2k+1$
vertices, such that any $\mathbb{S}^2$-embedding preserving
straight-line
 drawing of $G_k$
requires area $\Omega ( 2^k )$ under any resolution rule.
\end{theorem}

\begin{proof}
Firstly we orient the edges of $G_k$ to obtain an acyclic digraph.
Consider Fig.~\ref{fi:exponentialV2}(b). Note that $\ell_i$ is a
left vertex and $r_i$ is a right vertex for each $i$; we orient
edges from the left vertex to the right vertex. Also, we orient each
edge $(\ell_{i}, \ell_{i-1})$ from  $\ell_{i-1}$ to $\ell_{i}$ and
each edge $(r_i,r_{i-1})$ from $r_i$ to $r_{i-1}$ ($1 < i \leq k$).
%Finally, orient
%edge $(\ell_1, t)$ from $\ell_1$ to $t$, and
%edge $(t, r_1)$ from $t$ to $r_1$.

Let $\overrightarrow{G}_k$ denote the resulting directed acyclic
subgraph. Note that $\overrightarrow{G}_k$ is similar to the
$n$-vertex graph used by Di Battista et al.~\cite{DTT} to show that
upward planar drawings require $\Omega( 4^n )$ area; in fact
deletion of the edge $(s,t)$, and the vertices $s$ and $t$, yields
the Di Battista et al. graph.

Now suppose that the $\mathbb{S}^2$-topological graph $G_k$ has a
straight-line drawing $\Gamma_k$. Assume w.l.o.g. that edge $(s,t)$
is vertical in $\Gamma_k$, and every left vertex $\ell_i$ is drawn
in the left half-plane defined by the line through $(s,t)$.

The drawing $\Gamma_k$ has the same $\mathbb{S}^2$-embedding as
$G_k$, but may have a number of different $\mathbb{R}^2$-embeddings,
depending on the choice of the external face. Consider first the
simplest case, where $\Gamma_k$ has the $\mathbb{R}^2$-embedding as
shown in Fig.\ref{fi:exponentialV2}(b). That is, the external face
contains the vertices $s$, $\ell_k$, $r_k$, and $r_{k-1}$. We next
show that, in this case, $\Gamma_k$ is an upward planar drawing.

For every vertex $\ell_i$ ($1 < i < k$) in $\Gamma_k$, let $p$ be
the point where edge $(r_i, \ell_i)$ crosses edge $(s,t)$, and let
$q$ be the point where edge $(\ell_i, r_{i+1})$ crosses the line
through $(s,t)$. Similarly, let $p'$ be the point where edge
$(r_{i+1}, \ell_{i+1})$ crosses the line through the edge $(s,t)$,
and let $q'$ be the point where edge $(\ell_{i+1}, r_{i+2})$ crosses
the line through the edge $(s,t)$.

Since $\Gamma_k$ has the same $\mathbb{R}^2$-embedding of $G_k$ in
Fig.\ref{fi:exponentialV2}(b), the triangle $(\ell_i,p,q)$ is inside
the triangle $(\ell_{i+1},p',q')$, sharing only a portion of the
line through $s$ and $t$. Since this line is vertical, it follows
that the $x$-coordinate of $\ell_{i+1}$ is smaller than the
$x$-coordinate of $\ell_{i}$ in $\Gamma_k$. By a similar argument it
can be proved that the $x$-coordinate of $r_i$ must be smaller than
the $x$-coordinate of $r_{i+1}$ in $\Gamma_k$.

Hence the left-right subgraph has an upward drawing in $\Gamma_k$;
from the Theorem of Di Battista et al.~\cite{DTT}, the area is of
$\Gamma_k$ is $\Omega( 4^{2k+1} )$; this is $\Omega( 2^{k} )$.

We now consider the more general case, where the external face of
$\Gamma_k$ is not necessarily the same as in
Fig.\ref{fi:exponentialV2}(b). Since $\Gamma_k$ has the same
$\mathbb{S}^2$-embedding as $G_k$, the external face of $\Gamma_k$
should include one of the faces of the graph $\hat{G}_k$ formed from
$G_k$ by deleting the edge $(s,t)$. From Lemma~\ref{le:cycle}, the
external face cannot include an edge that crosses $(s,t)$; thus the
external face must be $F_j = (\ell_{j-1}, \ell_{j}, r_{j+1}, r_j)$
for some $1 < j < k$. As an example, the face $F_3$ is shown in
Fig.~\ref{fi:exponentialFaceF3}.
\begin{figure}[ht!]
  \centering
  \includegraphics[width=.9\columnwidth]{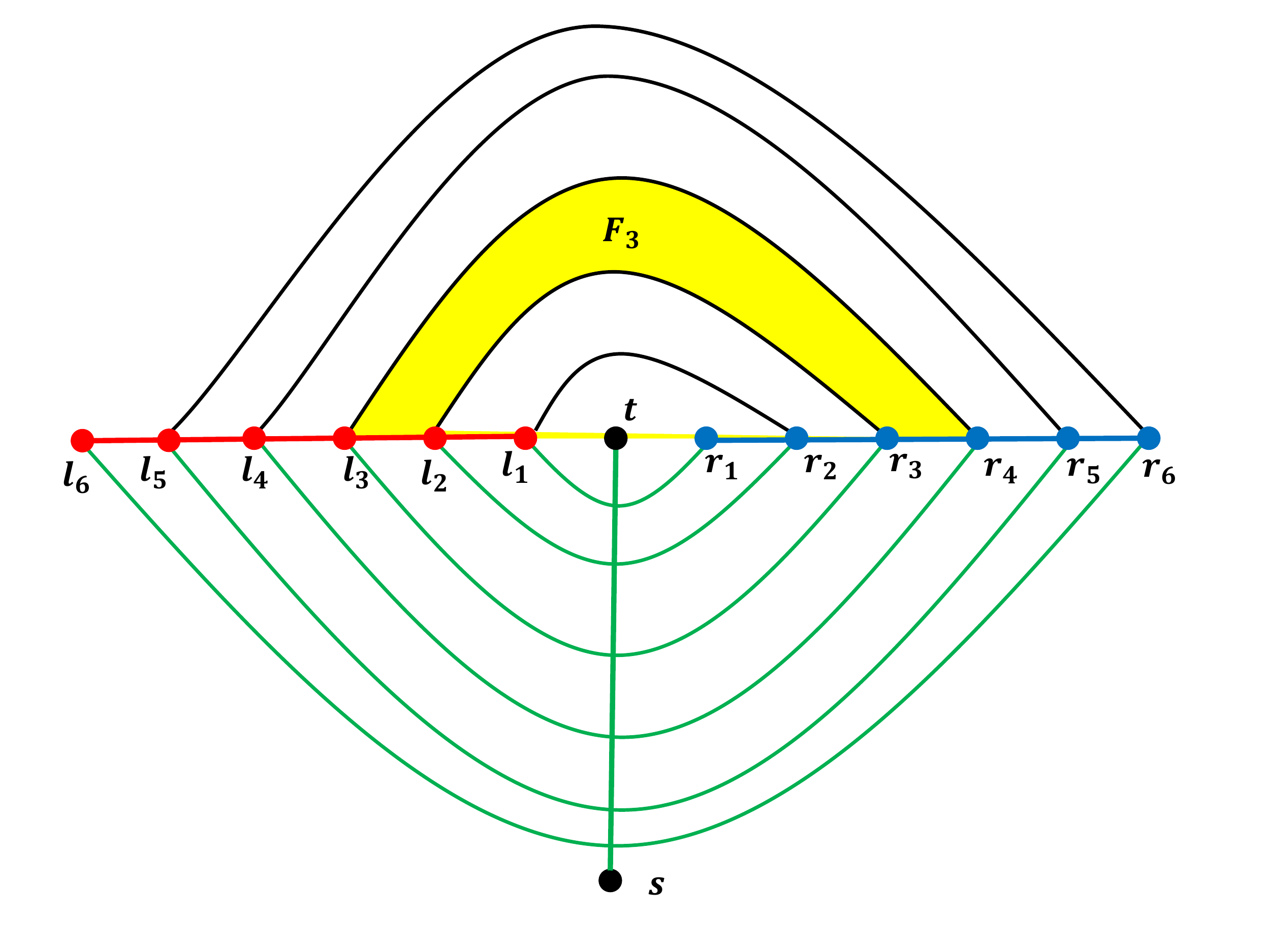}
  \caption{Face $F_3$ in $G_6$.}
  \label{fi:exponentialFaceF3}
\end{figure}

Now we consider two possibilities for the face $F_j$:
$j>\frac{k}{2}$, and $j \leq \frac{k}{2}$.

If $j>\frac{k}{2}$, then we consider the subgraph of $G_k$ induced
by $\{s,t \} \cup \{ \ell_j, \ell_{j-1} , \ldots , \ell_1 \} \cup \{ r_1 , r_2, \ldots
, r_j\}$. This is isomorphic to $G_j$, and its external face
contains $s$, $\ell_j$, $r_j$, and $r_{j-1}$. This corresponds to
the ``simplest case'' described above, and so any straight-line
drawing has area $\Omega ( 4^{2j+1} )$; since $j>\frac{k}{2}$, this
is $\Omega ( 2^k )$.

Next consider the case that $j \leq \frac{k}{2}$. In
Fig.~\ref{fi:exponentialF3Inverted} the $\mathbb{R}^2$-embedding
where $j=3$ for $G_6$ is illustrated.
\begin{figure}[ht!]
  \centering
  \includegraphics[width=.9\columnwidth]{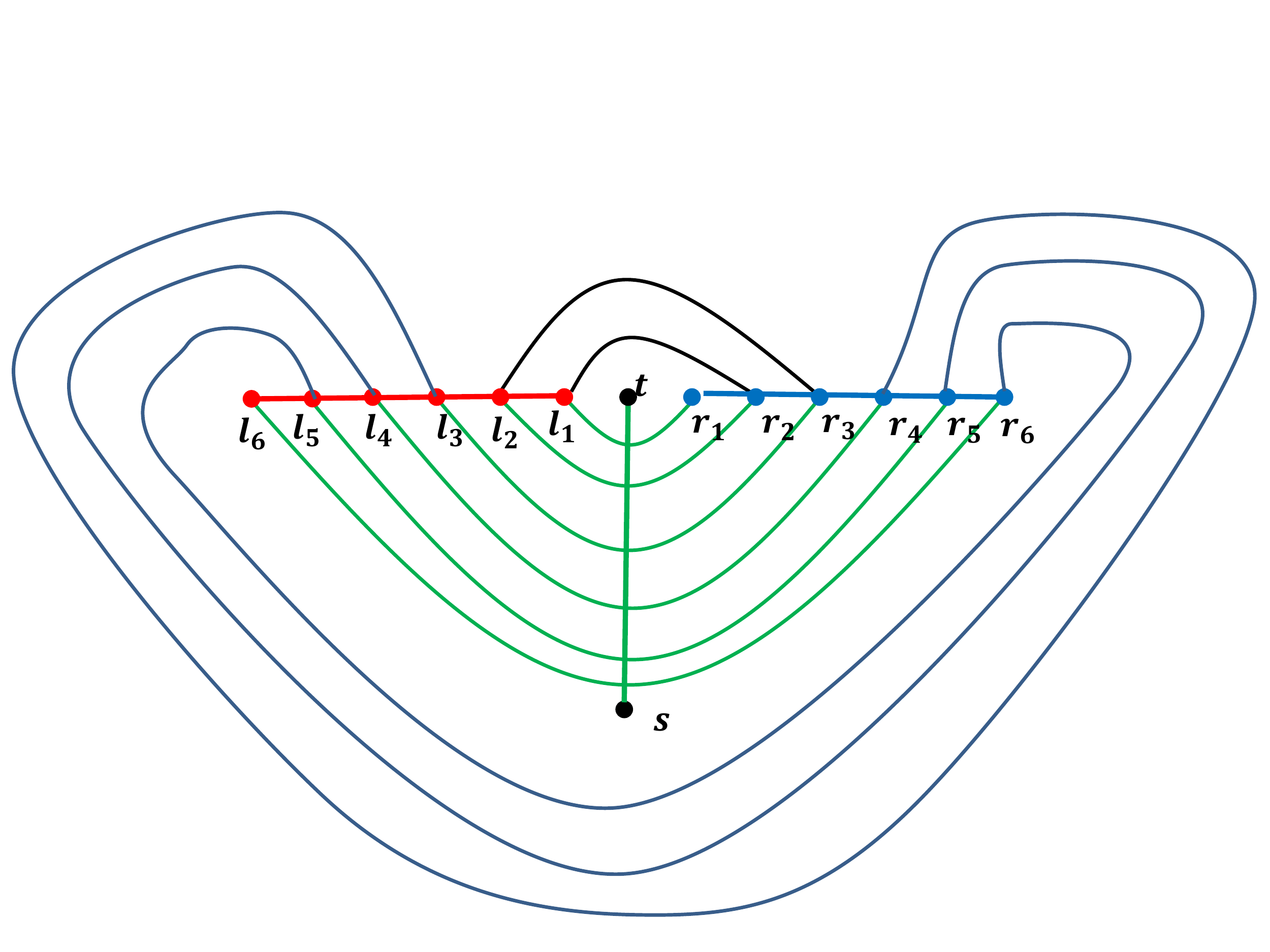}
  \caption{Using face $F_3$ from Fig.~\ref{fi:exponentialFaceF3} as the external face.}
  \label{fi:exponentialF3Inverted}
\end{figure}
In this case, consider the subgraph $H_{j,k}$ of $G_k$ induced by
$\{s,t, \ell_l , \ell_{k-1} , \ldots , \ell_j , r_j , r_{j+1},
\ldots , r_k\}$. The graph $H_{3,6}$ is illustrated in
Fig.~\ref{fi:exponentialH3}.
\begin{figure}[ht!]
  \centering
  \includegraphics[width=.9\columnwidth]{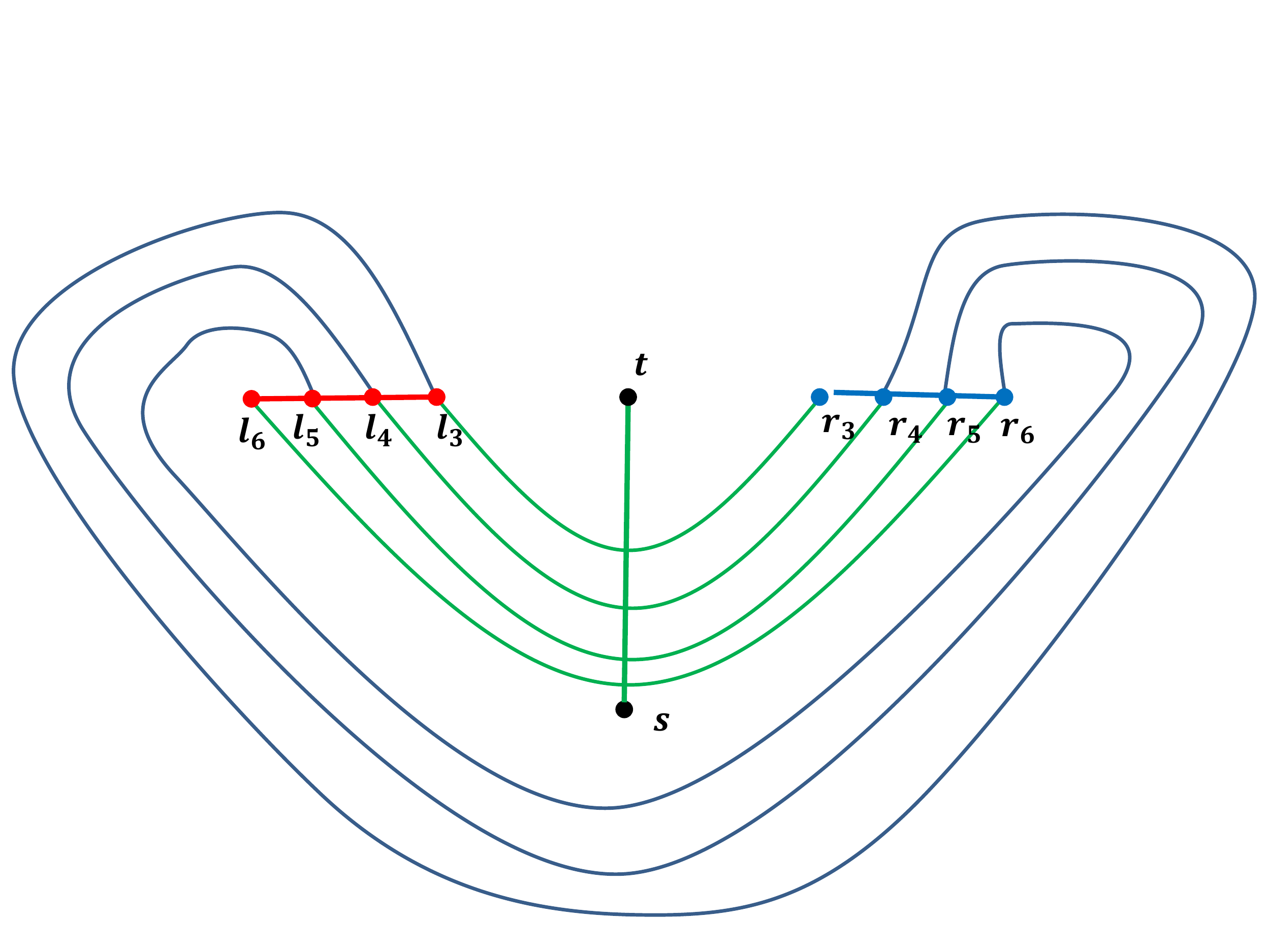}
  \caption{The graph $H_{3,6}$.}
  \label{fi:exponentialH3}
\end{figure}
In fact $H_{j,k}$ is isomorphic to $G_{k-j}$. The graph $H_{3,6}$ is
in Fig.~\ref{fi:exponentialH3} is re-drawn in
Fig.~\ref{fi:exponentialH3redrawn} to illustrate the isomorphism.
\begin{figure}[ht!]
  \centering
  \includegraphics[width=.9\columnwidth]{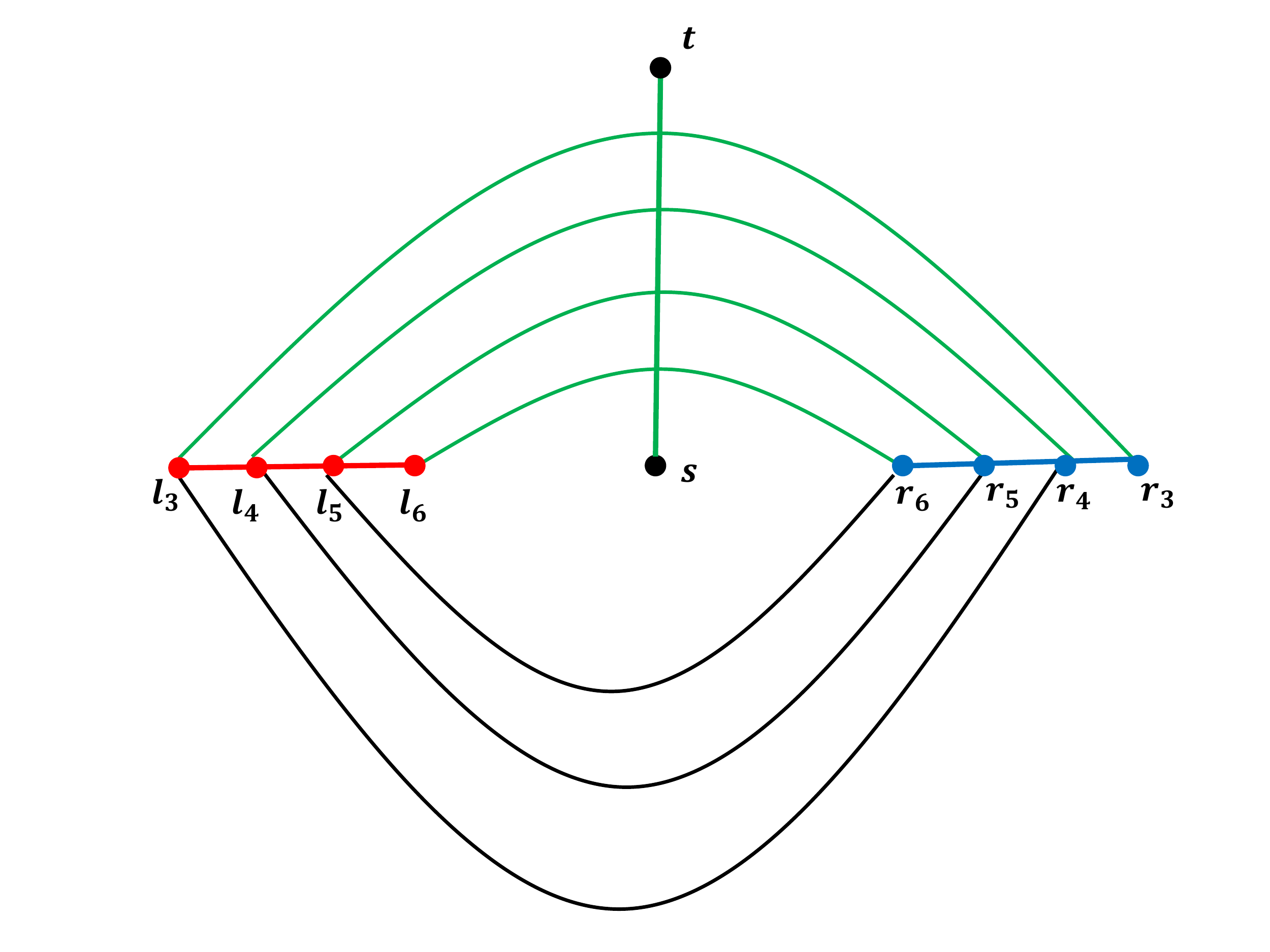}
  \caption{The graph $H_{3,6}$ redrawn to show the isomorphism with $G_3$.}
  \label{fi:exponentialH3redrawn}
\end{figure}
Thus, using the same argument as in the simple case above, we can
show that any straight-line drawing of $H_{j,k}$ has area $\Omega (
4^{2(k-j)+1} )$; since $j \leq \frac{k}{2}$, this is $\Omega ( 2^k
)$.

This completes the proof of Theorem~\ref{th:exponential}.
\end{proof}

We conclude this section by observing that the arguments used to
prove Theorem~\ref{th:s2} lead to a characterization of the maximal
almost-planar $\mathbb{R}^2$-topological graphs that have
$\mathbb{R}^2$-embedding preserving straight-line drawings.

\begin{theorem}\label{th:r2}
A maximal almost-planar $\mathbb{R}^2$-topological graph $G$
 admits an $\mathbb{R}^2$-embedding preserving straight-line
 drawing of $G$ if and only if every vertex
of $G$ is consistent, and every internal face of $G_{LR}$ is
consistent.
\end{theorem}

\begin{proof} The sufficiency of Theorem~\ref{th:r2} is proved in
Sections~\ref{ss:divide}, \ref{ss:outerGraph}, and
\ref{ss:innerGraph}. We show that the conditions of
Theorem~\ref{th:r2} are also necessary.

Suppose that $G_{LR}$ has an inconsistent face $f$, and has a
straight-line drawing $\Gamma$. Now $f$ contains a cycle $C$ with at
least one left vertex and at least one right vertex. From
Lemma~\ref{le:cycle}, $C$ contains $s$ and $t$, and the edge
$(s,t)$. Further, since $G_{LR}$ has no edge that joins a left
vertex to a right vertex, traversing $\Gamma$ in the clockwise
direction from $t$ gives a polygonal chain $C_L$ of vertices,
followed by $s$, followed by a polygonal chain $C_R$ of right
vertices.  All of $C_L$ is strictly to the left of the line through
$s$ and $t$, and $C_R$ is strictly right of this line. This is only
possible if the line segment between $s$ and $t$ lies inside $C$.
However, the edge $(s,t)$ cannot lie inside $C$ because $f$ is a
face of $G_{LR}$.
\end{proof}

\section{Open Problems}\label{se:open}

We mention two open problems that are naturally suggested by the
research in this paper. The first open problem is about
characterizing those almost-planar $\mathbb{R}^2$-topological graphs
that admit an embedding preserving straight-line drawing.
Theorem~\ref{th:r2} provides such a characterization for the family
of maximal almost-planar graphs.

The second open problem is about extending Theorem~\ref{th:s2} to
$k$-skew graphs with $k>1$. A topological graph $G=(V,E)$ is
$k$-skew if there is a set $E' \subset E$ of edges such that $G^- =
(V, E-E')$ has no crossings where $|E'| \leq k$. Many graphs that
arise in practice are $k$-skew for small values of $k$; this paper
gives drawing algorithms for the case $k=1$. For each edge $e \in
E'$, one could define ``left vertex relative to $e$'' and ``right
vertex relative to $e$'', extending the definitions of left and
right in this paper. However, it is not difficult to find a
topological 2-skew graph in which all vertices are consistent with
respect to the 2 ``crossing'' edges, but do not admit a
straight-line drawing. It would be interesting to characterize
$k$-skew graphs that admit a straight-line drawing for $k>1$.

\bibliographystyle{plain}
\bibliography{archiveX}

\end{document}